\documentclass[11pt]{article}
\usepackage{amssymb}
\usepackage{amsmath}
\usepackage{fullpage}
\usepackage{graphicx}
\usepackage{hyperref}
\usepackage{macros}
\usepackage{overpic}
\usepackage{color}
\usepackage{enumerate}
\usepackage[numbers]{natbib}

\definecolor{innerboxcolor}{rgb}{.9,.95,1}
\definecolor{outerlinecolor}{rgb}{.6,0,.2}

\newcommand{\optdomain}{\Theta}
\newcommand{\length}{L}
\newcommand{\fullmsg}{\ensuremath{\bar{\msg}}}

\newcommand{\Xsam}[1]{\ensuremath{X_{#1}}}
\newcommand{\thetahat}{\ensuremath{\widehat{\theta}}}
\newcommand{\real}{\ensuremath{\mathbb{R}}}
\newcommand{\defn}{\ensuremath{: \, =}}

\newcommand{\med}{\ensuremath{\operatorname{med}}}

\newcommand{\NORMAL}{\mathcal{N}}
\newcommand{\lamupper}{\ensuremath{\lambda_{\max}}}
\newcommand{\lamlower}{\ensuremath{\lambda_{\min}}}
\newcommand{\MiniMax}{\ensuremath{\mathfrak{M}}}

\makeatletter
\long\def\@makecaption#1#2{
  \vskip 0.8ex
  \setbox\@tempboxa\hbox{\small {\bf #1:} #2}
  \parindent 1.5em  
  \dimen0=\hsize
  \advance\dimen0 by -3em
  \ifdim \wd\@tempboxa >\dimen0
  \hbox to \hsize{
    \parindent 0em
    \hfil
    \parbox{\dimen0}{\def\baselinestretch{0.96}\small
      {\bf #1.} #2
    }
    \hfil}
  \else \hbox to \hsize{\hfil \box\@tempboxa \hfil}
  \fi
}
\makeatother

\begin{document}

\begin{center}
  \textbf{\LARGE Optimality guarantees for distributed \\ statistical
    estimation} \\
  \vspace{1cm}

  John C.\ Duchi$^1$  ~~~~
  Michael I.\ Jordan$^{1,2}$
  ~~~~ Martin J.\ Wainwright$^{1,2}$ ~~~~ Yuchen Zhang$^1$\\

  \vspace{.8cm}

  $^1$Department of Electrical Engineering and Computer Science \\
  and $^2$Department of Statistics \\
  University of California, Berkeley\\
  \texttt{\{jduchi,jordan,wainwrig,yuczhang\}@eecs.berkeley.edu} \\
  \vspace{.5cm}

  June 2014\footnote{An extended abstract of this manuscript~\cite{Zhang2013}
    appeared in the 2013 conference on
    \emph{Advances in Neural Information Processing Systems}.}
\end{center}

\begin{abstract}
  Large data sets often require performing distributed statistical estimation,
  with a full data set split across multiple
  machines and limited communication between machines. 
  To study such scenarios, we define and study
  some refinements of the classical minimax risk that apply to
  distributed settings, comparing to the performance of estimators
  with access to the entire data. Lower bounds on these quantities provide a
  precise characterization of the minimum amount of communication
  required to achieve the centralized minimax risk.  We study two
  classes of distributed protocols: one in which machines send messages
  independently over channels without feedback, and a second allowing
  for interactive communication, in which a central server broadcasts
  the messages from a given machine to all other machines.  We establish
  lower bounds for a variety of problems, including
  location estimation in several families and parameter estimation
  in different types of regression models.  Our results include a novel
  class of quantitative data-processing inequalities used to
  characterize the effects of limited communication.
\end{abstract}

\section{Introduction}

Rapid growth in the size and scale of datasets has fueled increasing
interest in statistical estimation in distributed settings (for
instance, see the papers~\cite{Boyd2011, Zhang2012, Dekel2012,
  Duchi2012, McDonald2010, Balcan2012} and references therein).
Modern data sets are often too large to be stored on a single computer,
and so it is natural to consider methods that involve multiple
computers, each assigned a smaller subset of the full dataset.  Yet
communication between machines or processors is often expensive, slow,
or power-intensive; as noted by \citeauthor{Fuller2011}
in a survey of the
future of computing, ``there is no known alternative to parallel
systems for sustaining growth in computing performance,'' yet the
power consumption and latency of communication is often relatively
high~\cite{Fuller2011}. Indeed, bandwidth limitations on network and inter-chip
communication often impose significant bottlenecks on algorithmic
efficiency. It is thus important to study the amount of communication
required between machines or chips in algorithmic development,
especially as we scale to larger and larger datasets.

The focus of the current paper is the communication complexity of a
few classes of statistical estimation problems.  Suppose we are
interested in estimating some parameter $\theta(\statprob)$ of an
unknown distribution $\statprob$, based on a dataset of $\totalobs$
i.i.d.\ observations.  In the classical setting, one considers
\emph{centralized estimators} that have access to all $\totalobs$
observations. One standard way to evaluate estimators is by studying
minimax rates of convergence, which characterize the optimal
(worst-case) performance over all centralized schemes.  By way of
contrast, in the distributed setting, one is given $\nummac$ different
machines, and each machine is assigned a subset of the sample of size
$\numobs = \lfloor \frac{\totalobs}{\nummac} \rfloor$.  Each machine
may perform arbitrary operations on its own subset of data, and it
then communicates results of these intermediate computations to the
other processors or to a central fusion node.  In this paper, we try
to answer the following question: what is the minimal number of bits
that must be exchanged in order to achieve (up to constant factors)
the optimal estimation error realized by a centralized scheme?

While there is a rich literature on statistical minimax theory
(e.g.~\cite{Ibragimov1981,Yu1997,Yang1999a,Tsybakov2009}), little
of it characterizes the effects of limiting communication.  In other
areas, ranging from theoretical computer
science~\cite{Yao1979,Abelson1980,Kushilevitz1997}, decentralized
detection and estimation (e.g.,~\cite{Tsitsiklis1993,Luo1994}), to
information theory (e.g.,~\cite{Han1998,ElGamal2011}), there is of
course a substantial literature on communication complexity.  Though
related to these bodies of work, our problem formulation and results
differ in several ways.
\begin{itemize}
\item In theoretical computer
  science~\cite{Yao1979,Abelson1980,Kushilevitz1997}, the prototypical
  problem is the distributed computation of a bivariate function
  $\theta : \mc{X} \times \mc{Y} \to \Theta$, defined on two discrete
  sets $\mc{X}$ and $\mc{Y}$, using a protocol that exchanges bits
  between processors.  The most classical problem is to find a
  protocol that computes $\theta(x, y)$ correctly for all $(x,y) \in
  \mc{X} \times \mc{Y}$, and exchanges the smallest number of bits to
  do so.  More recent work studies randomization and introduces
  information-theoretic measures for communication
  complexity~\cite{Chakrabarti2001,Bar-Yossef2004,Barak2010}, where
  the problem is to guarantee that $\theta(x, y)$ is computed
  correctly with high probability under a given (known) distribution
  $P$ on $x$ and $y$.  In contrast, our goal is to recover
  characteristics of an unknown distribution $P$ based on observations
  drawn from $P$.  Though this difference is somewhat subtle, it makes
  work on communication complexity difficult to apply in our
  settings. However, lower bounds on the estimation of population
  quantities $\theta(P)$ based on communication-constrained
  observations---including those we present here---do imply lower
  bounds in classical communication complexity settings.
\item Work in decentralized detection and estimation also studies
  limits of communication.  For example, \citet{Tsitsiklis1987}
  provide lower bounds on the difficulty of distributed convex
  optimization, and \citet{Luo1993} study limits on certain
  distributed algebraic computations. In these problems, as in other
  early work in communication complexity, data held by the distributed
  parties may be chosen adversarially, which precludes conclusions
  about statistical estimation. Other work in distributed control
  provides lower bounds on consensus and averaging, but in settings
  where messages sent are restricted to be of particular smooth
  forms~\cite{Olshevsky2009}.  Study of communication complexity has
  also given rise to interesting algorithmic schemes; for example,
  \citet{Luo2005} considers architectures in which machines may send
  only a single bit to a centralized processor; for certain problems,
  he shows that if each machine receives a single one-dimensional
  sample, it is possible to achieve the optimal centralized rate to
  within constant factors.
\item Information theorists have also studied problems of distributed
  estimation; for instance, see the paper of~\citet{Han1998} for an overview.
  In particular, this body of work focuses on the problem of testing a
  hypothesis or estimating a parameter from samples
  $\{(x_i,y_i)\}_{i=1}^\numobs$ where $\{(x_i)\}_{i=1}^\numobs$ and
  $\{(y_i)\}_{i=1}^\numobs$ are correlated but stored separately in two
  machines. \citet{Han1998} study estimation error for encoding rates $R >
  0$, or with sequences of rates $R_\numobs$ converging to zero as the
  sample size $\numobs$ increases.  In contrast to these asymptotic
  formulations---which often allow more communication than is required to
  attain centralized (unconstrained) minimax rates in our settings---we
  study fixed bounds on rates (say of the form $R_\numobs \leq
  c/\numobs$) for finite sample sizes $\numobs$, and we ask when it is possible
  to achieve the minimax statistical rate.
\end{itemize}

In this paper, we formulate and study two decentralized variants of
the centralized statistical minimax risk, one based on protocols that
engage in only a single round of message-passing, and the other based
on interactive protocols that can use multiple rounds of
communication.  The main question of interest is the following: how
must the communication budget $\budget$ scale as a function of
the sample size $\numobs$
at each machine, the total number of machines $\nummac$, and the
problem dimension $d$ so that the decentralized minimax risk matches
the centralized version up to constant factors? For some problems, we
exhibit an exponential gap between this communication requirement and
the number of bits required to describe the problem solution (up to
statistical precision); for instance, see
Theorems~\ref{theorem:gaussian-mean-communication}
and~\ref{theorem:interactive-gaussian} for results of this type. For
example, we show that for problems such as location estimation in Gaussian
and binomial families, the amount of communication must scale linearly
in the product $d \nummac$ of the dimension number of machines $\nummac$, which
is exponentially larger than the $\order(d \log \nummac)$ bits required
to specify the problem or communicate its solution. To
exhibit these gaps, we
provide lower bounds using information-theoretic techniques, with
the main novel ingredient being certain forms of quantitative data
processing inequalities. We also establish (nearly) sharp
upper bounds, some of which are based on recent work
by a subset of current authors on practical
schemes for distributed estimation (see~\citet{Zhang2012}).

The remainder of this paper is organized as follows.  We begin in
Section~\ref{sec:background} with background on the classical minimax
formalism, and then introduce two distributed variants of the minimax
risk.  Section~\ref{sec:mean-estimation} is devoted to the statement
of our main results, as well as the discussion of their consequences
in specific settings.  We turn to the proofs of our main results in
Section~\ref{sec:proof-sketch}, deferring more technical results to
the appendices, and we conclude in Section~\ref{SecDiscuss} with a
discussion.


\paragraph{Notation:}

For a random variable $X$, we let $P_X$ denote the probability measure
on $X$, so that $P_X(S) = P(X \in S)$, and we abuse notation by
writing $p_X$ for the probability mass function or density of $X$,
depending on the situation, so that $p_X(x) = P(X = x)$ in the
discrete case and denotes the density of $X$ at $x$ when $p_X$ is a
density.  We use $\log$ to denote $\log$-base $e$ and $\log_2$ for
$\log$ in base 2. For discrete random variable $X$, we let $H(X) =
-\sum_x p_X(x) \log p_X(x)$ denote the (Shannon) entropy (in ents),
and for probability distributions $P, Q$ on a set $\mc{X}$, with
densities $p, q$ with respect to a base measure $\mu$, we write the
KL-divergence as
\begin{equation*}
  \dkl{P}{Q} \defeq \int_\mc{X} p(x) \log \frac{p(x)}{q(x)} d\mu(x).
\end{equation*}
The mutual information $I(X; Y)$ between random variables $X$ and $Y$
where $Y$ has distribution $P_Y$ is defined as
\begin{align*}
  I(X; Y) \defeq \E_{P_X}\big[\dkl{P_Y(\cdot \mid X)}{P_Y(\cdot)}\big]
  = \int \dkl{P_Y(\cdot \mid X = x)}{P_Y(\cdot)} dP_X(x).
\end{align*}
We let $\vee$ and $\wedge$ denote maximum and minimum, respectively,
so that $a \vee b = \max\{a, b\}$. For an integer $k \ge 1$, we use
$[k]$ as shorthand for the set $\{1, \ldots, k\}$. We let
$a_{1:n}$ be shorthand for a sequence $a_1, \ldots, a_n$, and the notation
$a_n \gtrsim b_n$ means there is a numerical constant $c > 0$ such that
$a_n \ge c b_n$ for all $n$. Given a set
$A$, we let $\sigma(A)$ denote the Borel $\sigma$-field on $A$. 


\section{Background and problem formulation}
\label{sec:background}

In this section, we begin by giving background on the classical
notion of minimax risk in statistics.  We then introduce two
distributed variants of the minimax risk based on the notions of
independent and interactive protocols, respectively.

\subsection{Classical minimax risk}

For a family of probability distributions $\pclass$, consider a
function $\theta : \pclass \rightarrow \optdomain \subseteq \R^d$.  A
canonical example throughout the paper is the mean function, namely
$\theta(\statprob) = \E_\statprob[X]$.  Another simple example is the
median $\theta(\statprob) = \med_\statprob(X)$, or more generally,
quantiles of the distribution $\statprob$.  Now suppose that we are
given a collection of $\totalobs$ observations, say $X_{1:\totalobs} \defn
\{X_1, \ldots X_\totalobs\}$, drawn i.i.d. from some unknown member
$\statprob$ of $\pclass$.  Based on the sample $X_{1:\totalobs}$, our
goal is to estimate the parameter $\theta(\statprob)$, and an
estimator $\thetahat$ is a measurable function of the
$\totalobs$-vector $X_{1:\totalobs} \in \mc{X}^\totalobs$ into $\optdomain$.

We assess the quality of an estimator
$\thetahat = \thetahat(X_{1:\totalobs})$ via its mean-squared error
\begin{equation*}
  R(\thetahat, \theta(\statprob)) \defn \E_\statprob \big[ \|
  \thetahat(X_{1:\totalobs}) - \theta(\statprob) \|_2^2 \big],
\end{equation*}
where the expectation is taken over the sample $X_{1:\totalobs}$.  For
an estimator $\thetahat$, the function \mbox{$\statprob \mapsto
R(\thetahat, \theta(\statprob))$} defines the \emph{risk function} of
$\thetahat$ over the family $\pclass$.  Taking the supremum all
$\statprob \in \pclass$ yields the worst-case risk of the estimator.
The \emph{minimax rate} for the family $\pclass$ is defined in terms
of the best possible estimator for this worst-case criterion, namely
via the saddle point criterion
\begin{align}
  \label{EqnDefnMinimaxClassical}
  \MiniMax_\totalobs(\theta, \pclass) & \defn \inf_{\thetahat}
  \sup_{\statprob \in \pclass} R(\thetahat, \theta(\statprob)),
\end{align}
where the infimum ranges over all measurable functions of the data
$X_{1:\totalobs}$.  Many papers in mathematical statistics study the classical
minimax risk~\eqref{EqnDefnMinimaxClassical}, and its behavior is precisely
characterized for a range of
problems~\cite{Ibragimov1981,Yu1997,Yang1999a,Tsybakov2009}.  We
consider a few instances of such problems in the sequel.

\subsection{Distributed protocols}
\label{sec:distributed-protocols}

The classical minimax risk~\eqref{EqnDefnMinimaxClassical} imposes
no constraints on the
choice of estimator $\thetahat$.  In this section, we introduce a
refinement of the minimax risk that calibrates the effect of
communication constraints.  Suppose we have a collection of
$\nummac$ distinct computers or processing units.  Assuming for
simplicity\footnote{Although we assume in this paper that every
  machine has the same amount of data, our techniques are sufficiently
  general to allow for different sized subsets for each machine.  }
that $\totalobs$ is a multiple of $\nummac$, we can then divide our
full data set $X_{1:\totalobs}$ into a family of $\nummac$ subsets, each
containing $\numobs = \frac{\totalobs}{\nummac}$ distinct observations,
with $\Xsub{i}$ denoting the subset assigned to machine $i \in
[\nummac] = \{1, \ldots, \nummac \}$.  With this set-up, our goal is
to estimate $\theta(\statprob)$ via local operations
at each machine $i$ on the data subset $\Xsub{i}$ while performing
a limited amount of communication between machines.

More precisely, our focus is a class of distributed
protocols $\protocol$, in which at each round $t = 1, 2, \ldots$,
machine $i$ sends a message $\msg_{t,i}$ that is a measurable function
of the local data $\sample^{(i)}$ and potentially of past
messages. It is convenient to model this message as being sent to a
central fusion center.  Let $\fullmsg_t = \{\msg_{t,i}\}_{i \in
  [\nummac]}$ denote the collection of all messages sent at round $t$.
Given a total of $T$ rounds, the protocol $\protocol$ collects the
sequence $(\fullmsg_1, \ldots, \fullmsg_T)$, and constructs an
estimator \mbox{$\what{\theta} \defeq \what{\theta}(\fullmsg_1,
  \ldots, \fullmsg_T)$.} The length $\length_{t,i}$ of message
$\msg_{t,i}$ is the minimal number of bits required to encode it, and
the total length \mbox{$\length = \sum_{t = 1}^T \sum_{i=1}^\nummac
  \length_{t,i}$} of all messages sent corresponds to the \emph{total
  communication cost} of the protocol.  Note that the communication
cost is a random variable, since the length of the messages may depend
on the data, and the protocol may introduce auxiliary randomness.

It is useful to distinguish two different protocol classes, namely
\emph{independent} versus \emph{interactive}. An independent
protocol $\protocol$ is based on a single round ($T = 1$) of
communication in which machine $i$ sends a single message $\msg_{1,i}$ to the
fusion center.  Since there are no past messages, the message
$\msg_{1,i}$ can depend only on the local sample $\sample^{(i)}$.
Given a family $\pclass$, the class of independent protocols with
budget $\budget \geq 0$ is
\begin{equation*}
  \pind = \bigg\{ \mbox{ independent protocols $\protocol$ such that}
  ~ \sup_{\statprob\in\pclass}\E_\statprob\bigg[
    \sum_{i=1}^\nummac \length_i\bigg] \leq \budget\bigg\}.
\end{equation*}
(For simplicity, we use $\msg_i$ to indicate the message sent from processor
$i$ and $\length_i$ to denote its length in the independent case.)  It can be
useful in some situations to have more granular control on the amount of
communication, in particular by enforcing budgets on a per-machine basis.  In
such cases, we introduce the shorthand \mbox{$\budget_{1:\nummac} =
  (\budget_1, \ldots, \budget_\nummac)$} and define
\begin{equation*}
  \pindb = \big\{ \mbox{independent protocols $\protocol$ such that}
  ~ \sup_{\statprob\in\pclass}\E_\statprob[\length_i] \le
  \budget_i ~\mbox{for}~ i \in [\nummac]\big\}.
\end{equation*}

In contrast to independent protocols, the class of interactive protocols
allows for interaction at different stages of the message passing process.  In
particular, suppose that machine $i$ sends message $\msg_{t,i}$ to the fusion
center at time $t$, which then posts it on a ``public blackboard,'' where all
machines may read $\msg_{t,i}$ (this posting and reading incurs no
communication cost). We think of this as a global broadcast system, which may
be natural in settings in which processors have limited power or upstream
capacity, but the centralized fusion center can send messages without limit.
In the interactive setting, the message $\msg_{t,i}$ is a
measurable function of the local data $\sample^{(i)}$ and the past messages
$\fullmsg_{1:t-1}$.  The family of interactive protocols with budget $\budget
\ge 0$ is
\begin{equation*}
  \pinter = \bigg\{\mbox{interactive protocols $\protocol$ such that}
  ~ \sup_{P\in\pclass}\E_\statprob [ \length ] \leq
  \budget\bigg\}.
\end{equation*}

\subsection{Distributed minimax risks}

We can now define the distributed minimax risks that are
the central objects of study in this paper.  Our goal is to
characterize the best achievable performance of estimators $\thetahat$
that are functions of the vector of messages $\fullmsg_{1:T} \defn
(\fullmsg_1, \ldots, \fullmsg_T)$. As in the
classical minimax setting~\eqref{EqnDefnMinimaxClassical},
we measure the quality of a
protocol $\protocol$ and estimator $\thetahat$ by the mean-squared
error
\begin{align*}
  R(\thetahat, \theta(\statprob)) & \defn \E_{\statprob,\protocol} \big[
    \|\thetahat( \fullmsg_{1:T}) - \theta(\statprob)\|_2^2 \big],
\end{align*}
where the expectation is now taken over the randomness in the
messages, which is
due to both their dependence on the underlying data as well
as possible randomness in the protocol.  Given a communication budget
$\budget$, the \emph{minimax risk for independent protocols} is
\begin{equation}
  \label{eqn:minimax-risk}
  \minimaxind(\theta, \pclass, \budget) \defeq \inf_{\protocol \in
    \pind} \inf_{\what{\theta}} \sup_{\statprob \in \pclass} R \big(
  \thetahat, \theta(\statprob) \big).
\end{equation}
Here, the infimum is taken jointly over all independent procotols $\protocol$
that satisfy the budget constraint $\budget$, and over all estimators
$\what{\theta}$ that are measurable functions of the messages in the protocol.
The minimax risk~\eqref{eqn:minimax-risk} should also be understood to depend
on both the number of machines $\nummac$ and the individual sample size
$\numobs$ (we leave this implicit on the right hand side of
definition~\eqref{eqn:minimax-risk}). We
define the \emph{minimax risk for interactive
  protocols}, denoted by $\minimaxinter$, analogously, where we instead take the
infimum over the class of interactive protocols.  These
communication-dependent minimax risks are the central objects in this paper:
they provide a sharp characterization of the optimal estimation rate
as a function of the communication \mbox{budget $\budget$.}


\section{Main results and their consequences}
\label{sec:mean-estimation}

We now turn to the statement of our main results, along with some
discussion of their consequences.  We begin with a rather simple bound
based on the metric entropy of the parameter space; it confirms the
natural intuition that any procedure must communicate at least as many
bits as are required to describe a problem solution.  We show that
this bound is tight for certain problems,
but our subsequent more
refined techniques allow substantially sharper guarantees.

\subsection{Lower bound based on metric entropy}

We begin with a general but relatively naive lower bound that depends
only on the geometric structure of the parameter space, as captured by
its metric entropy.  In particular, given a subset $\optdomain \subset
\R^d$, we say $\{\theta^1, \ldots, \theta^K\}$ are $\delta$-separated
if $\ltwos{\theta^i - \theta^j} \ge \delta$ for $i \neq j$. We then
define the \emph{packing entropy} of $\optdomain$ as
\begin{equation*}
  \pentropy_\optdomain(\delta) \defeq \log_2 \max \big\{M \in \N \mid
  \{\theta_1, \ldots, \theta_M\} \subset \optdomain ~
  \mbox{are~$\delta$-separated}\big\}.
\end{equation*}
It is straightforward to see that the packing entropy continuous from
the right and non-increasing in $\delta$, so that the inverse function
$\pentropy_\optdomain^{-1}(B) \defeq \sup \{\delta \mid
\pentropy_\optdomain(\delta) \ge B\}$ is well-defined.  With this
definition, we have the following claim:
\begin{proposition}
  \label{PropGeneral}
  For any family of distributions $\mc{P}$ and parameter set $\Theta =
  \theta(\mc{P})$, the interactive minimax risk is lower bounded as
  \begin{align*}
    \minimaxinter(\theta, \pclass, \budget) & \geq \frac{1}{8}
    \big(\pentropy^{-1}_\optdomain(2 \budget + 2) \big)^2.
  \end{align*}
\end{proposition}
\noindent We prove this proposition in
Section~\ref{sec:proof-prop-general}.  The same lower bound
trivially holds for $\minimaxind(\theta, \pclass, \budget)$, as any
independent protocol is a special case of an interactive protocol.
Although Proposition~\ref{PropGeneral} is a relatively generic
statement, not exploiting any particular structure of the problem, it
is in general unimprovable by more than constant factors, as the
following example illustrates.

\begin{exas}[Bounded mean estimation]
  Suppose our goal is to estimate the mean $\theta = \theta(P)$ of
  a class of distributions $\pclass$ supported on the interval $[0, 1]$,
  so that $\optdomain = \theta(\pclass) = [0, 1]$.  Suppose that a
  single machine ($\nummac = 1$) receives $\numobs$ i.i.d.\ observations
  $\Xsam{i}$ according to $P$.  The packing
  entropy has lower bound $\pentropy_\optdomain(\delta) \geq
  \log_2(1/\delta)$, and consequently, Proposition~\ref{PropGeneral} implies
  that the distributed minimax risk is lower bounded as
  \begin{align*}
    \minimaxind(\theta, \pclass, \budget) & \geq \minimaxinter(\theta,
    \pclass, \budget) \; \geq \; \frac{1}{8} \big(2^{ - 2
      \budget - 2} \big)^2.
  \end{align*}
  Setting $\budget = \frac{1}{4} \log_2 \numobs$ yields the lower bound
  $\minimaxind(\theta, \pclass([0, 1]), \budget) \geq \frac{1}{128
    \numobs}$.

  This lower bound is sharp up to the constant pre-factor; it can
  be achieved by the following simple method.  Given its $\numobs$
  observations, the single machine computes the sample mean
  \mbox{$\widebar{X}_\numobs \defn \frac{1}{\numobs} \sum_{i=1}^\numobs
    \Xsam{i}$.}  The sample mean must lie in the interval $[0, 1]$, and
  so can be quantized to accuracy $\frac{1}{\numobs}$ using
  $\log_2\numobs$ bits, and this quantized version $\thetahat$ can be
  transmitted.  A straightforward calculation shows that $\E[(\thetahat
    - \theta)^2] \leq \frac{2}{\numobs}$, and
  Proposition~\ref{PropGeneral} yields an order-optimal bound.
  \hfill $\blacksquare$
\end{exas}


\subsection{Independent protocols in multi-machine settings}

We would like to study how the \emph{budget $\budget$}---the number of bits
required to achieve the minimax rate---scales with the number of
machines $\nummac$.  For our first set of results in this setting, we
consider the non-interactive case, where each machine $i$ sends
messages $\msg_i$ independently of all the other machines.  These
results serve as pre-cursors to our later results on interactive
protocols.

We first provide lower bounds for mean estimation in the
$d$-dimensional normal location family model:
\begin{align}
  \label{EqnMultNormalLoc}
  \NORMAL_d & \defn \{\normal(\theta, \sigma^2 I_{d \times d}) \mid
  \theta \in \optdomain = [-1, 1]^d\}.
\end{align}
Here each machine receives an i.i.d.\ sample of size $\numobs$ from a
normal distribution $\normal(\theta, \sigma^2 I_{d\times d})$ with
unknown mean $\theta$. The following result provides a lower bound on
the distributed minimax risk with independent communication:
\begin{theorem}
  \label{theorem:gaussian-mean-communication}
  Given a communication budget $\budget_i$ for each machine $i = 1,
  \ldots, \nummac$, there exists a universal (numerical) constant $c$
  such that
  \begin{align}
    \label{eqn:gaussian-mean-communication}
    \minimaxind(\theta, \NORMAL_d, \budget_{1:\nummac})
    & \geq c \frac{\sigma^2
      d}{\nummac \numobs} \; \min \Big \{\frac{\nummac
      \numobs}{\sigma^2}, \; \frac{\nummac}{\log \nummac}, \; \frac{
      \nummac }{(\sum_{i=1}^\nummac \min \{1, \frac{\budget_i}{d} \} ) \,
      \log \nummac } \vee 1 \Big\}.
  \end{align}
\end{theorem}
\noindent See
Section~\ref{appendix:proof-gaussian-mean-communication} for the proof
of this claim. \\

Given centralized access to the full $\nummac \numobs$-sized sample,
the minimax rate for the mean-squared error is $\frac{\sigma^2
  d}{\nummac \numobs}$ (e.g.~\citet{Lehmann1998}).  This optimal
rate is achieved by the sample mean.  Consequently, the lower
bound~\eqref{eqn:gaussian-mean-communication} shows that each machine
individually must communicate at least $\frac{d}{\log \nummac}$ bits
for a decentralized procedure to match the centralized rate.  If we
ignore logarithmic factors, this lower bound is achievable by the
following simple procedure:
\begin{enumerate}[(i)]
\item First, each machine computes the sample mean of its local data
  and truncates it to the interval $[-1 -
    \frac{\sigma}{\sqrt{\numobs}}, 1 +
    \frac{\sigma}{\sqrt{\numobs}}]$.
\item Next, each machine quantizes each coordinate of the resulting
  estimate to precision $\frac{\sigma^2}{\nummac\numobs}$, using $\order(1)
  \,d \log \frac{\nummac\numobs}{\sigma^2} $ bits to do so.
\item The machines send these quantized averages to the fusion
  center using \mbox{$\budget = \order(1)\, d \nummac
    \log \frac{\numobs}{\sigma^2}$} total bits.
\item Finally, the fusion center averages them, obtaining an estimate
  with mean-squared error of the order $\frac{\sigma^2 d}{\nummac \numobs}$.
\end{enumerate}


The techniques we develop also apply to other families of probability
distributions, and we finish our discussion of independent
communication protocols by presenting a result that
gives lower bounds sharp to numerical constant prefactors.
In particular, we consider mean estimation for the family $\mc{P}_d$
of distributions supported on the compact set $[-1, 1]^d$.  One
instance of such a distribution is the Bernoulli family taking values
on the Boolean hypercube $\{-1, 1\}^d$.
\begin{proposition}
  \label{proposition:single-bernoulli-samples}
  Assume that each of $\nummac$ machines receives a single observation
  ($\numobs = 1$) from a distribution in $\mc{P}_d$.  There exists a
  universal constant $c > 0$ such that
  \begin{align*}
    \minimaxind(\theta, \pclass_d, \budget_{1:\nummac}) & \geq c \,
    \frac{d}{\nummac} \; \min \Big\{ \nummac, \; \frac{\nummac}{ \sum_{i
        = 1}^\nummac \min\{1, \frac{\budget_i}{d} \} } \Big\},
  \end{align*}
  where $\budget_i$ is the budget for machine $i$.
\end{proposition}
\noindent
\noindent See Section~\ref{appendix:proof-theorem-single-bernoulli}
for the proof. \\

The standard minimax rate for $d$-dimensional mean estimation on
$\mc{P}_d$ scales as $d /\nummac$, which is achieved by the sample
mean. Proposition~\ref{proposition:single-bernoulli-samples}
shows that to
achieve this scaling, we must have $\sum_{i=1}^\nummac \min \{1,
\frac{\budget_i}{d}\} \gtrsim \nummac$, showing that each machine must
send $\budget_i \gtrsim d$ bits.  This lower bound is also achieved by
a simple scheme:
\begin{enumerate}[(i)]
\item Each machine $i$ receives an observation $X_i \in [-1, 1]^d$.
  Based on this observation.  it generates a Bernoulli random vector
  \mbox{$Z_i = (Z_{i1}, \ldots, Z_{id})$} with $Z_{ij} \in \{0,1\}$
  taking the value $1$ with probability $(1 + \sample_{ij}) / 2$,
  independently across coordinates.
\item Machine $i$ uses $d$ bits to send the vector $Z_i \in \{0,1\}^d$
  to the fusion center.
\item The fusion center then computes the average $\thetahat =
  \frac{1}{\nummac} \sum_{i=1}^\nummac (2 Z_i - 1)$.  This average is
  unbiased, and its expected squared error is bounded by $d /
  \nummac$.
\end{enumerate}

Note that for both the normal location family of
Theorem~\ref{theorem:gaussian-mean-communication} and the simpler
bounded single observation model in
Proposition~\ref{proposition:single-bernoulli-samples}, there is an
exponential gap between the information required to describe the
problem to the minimax mean squared error of $\frac{d}{\nummac
  \numobs}$---which scales as as $\order(1) d \log(\nummac \numobs)$---and
the number of bits that must be communicated, which scales nearly
linearly in $\nummac$.  See also our discussion following
Theorem~\ref{theorem:interactive-gaussian}.


\subsection{Interactive protocols in multi-machine settings}

Having provided results on mean estimation in the non-interactive
setting, we now turn to the substantially harder setting of
distributed statistical inference where feedback is permitted.  As
described in Section~\ref{sec:distributed-protocols},
in the interactive setting the fusion center may
freely broadcast every message received to all other
machines in the network.  This freedom allows more powerful
algorithms, rendering the task of proving lower bounds more challenging.

Let us begin by considering the uniform location family $\mc{U}_d
= \{ P_{\theta}, \; \theta \in [-1, 1]^d \}$, where $P_\theta$ is the
uniform distribution on the rectangle $[\theta_1 - 1, \theta_1 + 1]
\times \cdots \times [\theta_d - 1, \theta_d + 1]$.  For this problem,
a direct application of Proposition~\ref{PropGeneral} gives a nearly
sharp result:
\begin{proposition}
  \label{corollary:uniform-upper-bound}
  Consider the uniform location family $\mathcal{U}_d$ with $\numobs$
  i.i.d.\ observations per machine:
  \begin{enumerate}
  \item[(a)] There are universal (numerical)
    constants
    $c_1, c_2 > 0$ such that
    \begin{align*}
      \minimaxinter(\theta, \mathcal{U}, \budget) \geq
      c_1 \max\left\{ \exp\left(-c_2 \frac{\budget}{d}\right),
      \frac{d}{(\nummac \numobs)^2}\right\}.
    \end{align*}
    \vspace{-15pt}
  \item[(b)] Conversely, given a budget of $\budget = d\big[2 \log_2(2
    \nummac \numobs) + \log(\nummac)(\ceil{\log_2 d} + 2 \log_2(2
    \nummac\numobs))\big]$ bits, there is a universal constant $c$ such that
    \begin{align*}
      \minimaxinter(\theta, \mathcal{U}, \budget) \leq c \frac{d}{(
        \nummac \numobs )^2}.
    \end{align*}
  \end{enumerate}
\end{proposition}
\noindent See Section~\ref{sec:proof-uniform} for the proof of this
claim. \\

If each of the $\nummac$ machines receives $\numobs$ observations, we
have a total sample size of $\nummac \numobs$, so the minimax rate
over all centralized procedures scales as $d/(\nummac \numobs)^2$ (for
instance, see~\citet{Lehmann1998}). Consequently,
Proposition~\ref{corollary:uniform-upper-bound}(b) shows that the
number of bits required to achieve the centralized rate has only
\emph{logarithmic} dependence on the number $\nummac$ of machines.
Part (a) shows that this logarithmic dependence on $\nummac$ is
unavoidable: at least $\budget \gtrsim d \log(\nummac \numobs)$ bits
are necessary to attain the optimal rate of $\frac{d}{(\nummac
  \numobs)^2}$.

It is natural to wonder whether such logarithmic dependence holds more
generally. The following result shows that it does not: for some
problems, the dependence on $\nummac$ must be (nearly) linear.  In
particular, we reconsider estimation in the normal location family
model~\eqref{EqnMultNormalLoc}, showing a lower bound that is nearly
identical to that of
Theorem~\ref{theorem:gaussian-mean-communication}.
\begin{theorem}
  \label{theorem:interactive-gaussian}
  For $i = 1, \ldots, \nummac$, assume that each machine receives an
  i.i.d.\ sample of size $n$ from a normal location
  model~\eqref{EqnMultNormalLoc} and that there is a total
  communication budget $\budget$. Then there exists a universal (numerical)
  constant $c$ such that
  \begin{equation}
    \minimaxinter(\theta, \NORMAL_d, \budget) \geq c
    \frac{\sigma^2 d}{\nummac \numobs} \; \min\Big\{\frac{\nummac
      \numobs}{\sigma^2}, \; \frac{
      \nummac }{(\budget/d+1) \; \log \nummac } \vee 1\Big\}.
    \label{eqn:interactive-gaussian}
  \end{equation}
\end{theorem}
\noindent See Section~\ref{sec:proof-interactive-gaussian} for the proof
of this claim. \\

Theorem~\ref{theorem:interactive-gaussian} is analogous to, but
slightly weaker than, the corresponding lower bound from
Theorem~\ref{theorem:gaussian-mean-communication} for the
non-interactive setting.  In particular, the lower
bound~\eqref{eqn:interactive-gaussian} shows that at least $\budget
\gtrsim \frac{d \nummac}{\log \nummac} $ bits are required for any
distributed procedure---even allowing fully interactive
communication---to attain the centralized minimax rate.  Thus, in
order to achieve the minimax rate up to logarithmic factors, the total
number of bits communicated must scale (nearly) linearly with the product of
the dimension $d$ and number of machines $\nummac$.

Moreover, these two theorems show that there is an exponential gap between the
number of bits required to communicate the problem solution and the number
required to compute it in a distributed manner.  More specifically, assuming
(for simplicity) that $\sigma^2 = 1$, describing a solution of the normal mean
estimation problem to accuracy $\frac{d}{\nummac \numobs}$ in squared
$\ell_2$-error requires at most $\order(1) d \log(\nummac \numobs)$ bits.  On
the other hand, these two theorems show that nearly $d \nummac$ bits
\emph{must} be communicated. This linear scaling in $\nummac$ is dramatically
different from---exponentially worse than---the logarithmic scaling for the
uniform family. Establishing sharp communication-based lower bounds
thus requires careful study of the underlying family of distributions.

Note that in both Theorems~\ref{theorem:gaussian-mean-communication}
and~\ref{theorem:interactive-gaussian}, the upper and lower bounds
differ by logarithmic factors in the sample size $\numobs$ and number
of machines $\nummac$.  It would be interesting to close this minor
gap.  Another open question is whether the distributed minimax rates
for the independent and interactive settings are the same up to
constant factors, or whether their scaling actually differs in terms
of these logarithmic factors.


\subsection{Consequences for regression}
\label{sec:optimization}

The problems of mean estimation studied in the previous section,
though simple in appearance, are closely related to other, more
complex problems.  In this section, we show how lower bounds on mean
estimation can be used to establish lower bounds for distributed
estimation in two standard but important generalized linear
models~\cite{Hastie1995}: linear regression and probit regression.


\subsubsection{Linear regression}
\label{sec:linear-regression-results}

Let us begin with a distributed instantiation of linear regression
with fixed design matrices.  Concretely, suppose that each of
$\nummac$ machines has stored a fixed design matrix $\design^{(i)} \in
\real^{\numobs \times d}$ and then observes a response vector
$\response^{(i)} \in \R^d$ from the standard linear regression model
\begin{align}
  \label{eqn:linear-regression-model}
  \response^{(i)} = \design^{(i)} \theta + \noise^{(i)},
\end{align}
where $\noise^{(i)} \sim \normal(0,\sigma^2I_{\numobs\times\numobs})$
are independent noise vectors.  Our goal is to estimate the unknown regression
vector $\theta \in \Theta = [-1,1]^d$, identical for each machine.
Our result involves the smallest and largest
eigenvalues of the rescaled design matrices via the quantities
\begin{align}
  \label{EqnEigBound}
  \lamupper^2 \; \defeq \max_{i \in \{1, \ldots, \nummac\}}
  \frac{\eigenmax({\design^{(i)}}^\top \design^{(i)})}{\numobs},
  \quad \mbox{and} \quad
  \lamlower^2 \defeq \min_{i \in \{1, \ldots, \nummac\}}
  \frac{\eigenmin({\design^{(i)}}^\top \design^{(i)})}{\numobs} > 0.
\end{align}

\begin{corollary}
  \label{corollary:linear-regression-bound}
  Given the linear regression model~\eqref{eqn:linear-regression-model},
  there is a universal positive constant $c$ such that
  \begin{subequations}
    \begin{align}
      \label{EqnLinearRegressionLower}
      \minimaxinter(\theta, \pclass, \budget) & \geq c
      \frac{\sigma^2 d}{\lamupper^2 \nummac \numobs} \; \min \bigg\{
      \frac{\lamupper^2 \nummac \numobs}{\sigma^2}, \; \frac{ \nummac
      }{(\budget/d + 1) \; \log \nummac } \vee 1 \bigg\}.
    \end{align}
    Conversely, given a budgets $\budget_i \geq d \nummac \log(\nummac
    \numobs)$, there is a universal constant $c'$ such that
    \begin{align}
      \label{EqnLinearRegressionUpper}
      \minimaxind(\theta, \pclass, \budget_{1:\nummac}) \leq
      \frac{c'}{\lamlower^2} \; \frac{\sigma^2 d}{\nummac \numobs}.
    \end{align}
  \end{subequations}
\end{corollary}

It is a classical fact (e.g.~\cite{Lehmann1998}) that the minimax rate
for $d$-dimensional linear regression scales as $d\sigma^2/(\numobs
\nummac)$.  Part (a) of
Corollary~\ref{corollary:linear-regression-bound} shows this optimal
rate is attainable only if the total budget $\budget$ grows as
$\frac{d \nummac}{\log \nummac}$.  Part (b) of the corollary shows
that the minimax rate is achievable---even using an independent
protocol---with budgets that match the lower
bound to within logarithmic factors.

\begin{proof}
  The upper bound~\eqref{EqnLinearRegressionUpper} follows from the
  results of \citet{Zhang2012}.  Their results imply that the upper
  bound can be achieved by solving each regression problem separately,
  quantizing the (local) solution vectors \mbox{$\what{\theta}^{(i)} \in
    [-1, 1]^d$} to accuracy $\frac{1}{\nummac \numobs}$ using $\budget_i
  = \ceil{d \log_2(\nummac \numobs)}$ bits and performing a form of
  approximate averaging.

  In order to prove the lower bound~\eqref{EqnLinearRegressionLower}, we
  show that solving an arbitrary Gaussian mean estimation problem can be
  reduced to solving a specially constructed linear regression problem.
  This reduction allows us to apply the lower bound from
  Theorem~\ref{theorem:interactive-gaussian}.  Given $\theta \in
  \Theta$, consider the Gaussian mean model
  \begin{align*}
    X^{(i)} = \theta + w^{(i)}, ~~ \mbox{where} ~~ w^{(i)} \sim
    \normal\Big(0, \frac{\sigma^2}{\lamupper^2 \numobs} I_{d \times
      d}\Big).
  \end{align*}
  Each machine $i$ has its own design matrix $A^{(i)}$, and we use it to
  construct a response vector \mbox{$\response^{(i)} \in \R^\numobs$.}
  Since $\eigenmax({\design^{(i)}}^\top \design^{(i)}/\numobs) \le \lamupper^2$, the
  matrix $\Sigma^{(i)} \defeq \sigma^2I_{\numobs\times\numobs} -
  \frac{\sigma^2}{\lamupper^2 n}\design^{(i)}(\design^{(i)})^\top$ is
  positive semidefinite.  Consequently, we may form a response vector
  via
  \begin{align}
    \label{eqn:construct-gaussian-noise}
    \response^{(i)} = \design^{(i)}X^{(i)} + z^{(i)} = \design^{(i)}\theta
    + \design^{(i)} w^{(i)} + z^{(i)} , \quad \mbox{ $z^{(i)} \sim
      \normal(0, \Sigma^{(i)})$ independent of $w^{(i)}$.}
  \end{align}
  The independence of $w^{(i)}$ and $z^{(i)}$ guarantees that
  $\response^{(i)} \sim \normal(\design^{(i)}\theta, \sigma^2
  I_{\numobs\times\numobs})$, so the pair $(\response^{(i)},
  \design^{(i)})$ is faithful to the regression
  model~\eqref{eqn:linear-regression-model}.

  Now consider a protocol $\protocol\in \pinter$ that can solve any
  regression problem to within accuracy $\delta$, so that
  $\E[\ltwos{\what{\theta} - \theta}^2]\le \delta^2$.  By the previously
  described reduction, the protocol $\protocol$ can also solve the mean
  estimation problem to accuracy $\delta$, in particular via the pair
  $(\design^{(i)},\response^{(i)})$ described in
  expression~\eqref{eqn:construct-gaussian-noise}.  Combined with this
  reduction, the corollary thus follows from
  Theorem~\ref{theorem:interactive-gaussian}.
\end{proof}


\subsubsection{Probit regression}

We now turn to the problem of binary classification, in particular
considering the probit regression model.  As in the previous section,
each of $\nummac$ machines has a fixed design matrix
\mbox{$\design^{(i)}\in \R^{\numobs\times d}$,} where $\design^{(i,
  k)}$ denotes the $k$th row of $\design^{(i)}$.  Machine $i$ receives
$\numobs$ binary responses \mbox{$Z^{(i)} =
  (Z^{(i,1)},\dots,Z^{(i,\numobs)})$}, drawn from the conditional
distribution
\begin{equation}
  \label{eqn:probit-model}
  \P \big( Z^{(i,k)} = 1 \mid \design^{(i,k)}, \theta \big) =
  \probit(\design^{(i,k)}\theta) ~~~~ \mbox{for~some~fixed~} \theta
  \in \Theta = [-1, 1]^d,
\end{equation}
where $\probit(\cdot)$ denotes the standard normal CDF.  The
log-likelihood of the probit model~\eqref{eqn:probit-model} is concave
(cf.~\cite[Exercise 3.54]{Boyd2004}).  Under
condition~\eqref{EqnEigBound} on the design matrices, we have:
\begin{corollary}
  \label{corollary:probit-regression-bound}
  Given the probit model~\eqref{eqn:probit-model}, there
  is a  universal constant $c > 0$ such that
  \begin{subequations}
    \begin{align}
      \label{EqnProbitLower}
      \minimaxinter(\theta, \pclass, \budget_{1:\nummac}) & \geq c
      \frac{d}{\lamupper^2 \nummac \numobs} \; \min \Big \{\lamupper^2
      \nummac \numobs, \; \frac{ \nummac }{(\budget/d + 1) \; \log \nummac
      } \Big\}.
    \end{align}
    Conversely, given a budgets $\budget_i \geq d \log(\nummac
    \numobs)$, there is a universal constant $c'$ such that
    \begin{align}
      \label{EqnProbitUpper}
      \minimaxind(\theta, \pclass, \budget_{1:\nummac}) \leq \frac{c'}{
        \lamlower^2 }\; \frac{d}{\nummac \numobs}.
    \end{align}
  \end{subequations}
\end{corollary}


\begin{proof}
  As in Corollary~\ref{corollary:linear-regression-bound}, the upper
  bound~\eqref{EqnProbitUpper} follows from the results
  of~\citet{Zhang2012}.

  Turning to the lower bound~\eqref{EqnProbitLower}, our strategy is to
  show that probit regression is at least as hard as linear regression,
  in particular by demonstrating that any linear regression problem can
  be solved via estimation in a specially constructed probit model.
  Given an arbitrary regression vector $\theta \in \optdomain$, consider
  a linear regresion problem~\eqref{eqn:linear-regression-model} with
  noise variance $\sigma^2 = 1$.  We construct the binary responses for
  our probit regression $(Z^{(i,1)},\dots,Z^{(i,\numobs)})$ by
  \begin{align}
    \label{eqn:construct-probit-response}
    Z^{(i,k)} = \begin{cases} 1 & \mbox{if~} \response^{(i,k)}\geq 0,\\ 0
      & \mbox{otherwise}.
    \end{cases}
  \end{align}
  By construction, we have $\P(Z^{(i,k)}=1 \mid \design^{(i)}, \theta) =
  \probit(\design^{(i,k)}\theta)$ as desired for our
  model~\eqref{eqn:probit-model}.  By inspection, any protocol
  $\protocol\in\pinter$ solving the probit regression problem provides
  an estimator with the same mean-squared error as the original linear
  regression problem via the
  construction~\eqref{eqn:construct-probit-response}.  Consequently, the
  lower bound~\eqref{EqnProbitLower} follows from
  Corollary~\ref{corollary:linear-regression-bound}.
\end{proof}



\section{Proofs}
\label{sec:proof-sketch}

We now turn to the proofs of our main results, deferring
more technical results to the appendices.

\subsection{Proof of Proposition~\ref{PropGeneral}}
\label{sec:proof-prop-general}

This result is based on the classical reduction from estimation to
testing (e.g.,~\cite{LeCam1973,Ibragimov1981,Birge1983}).  For a given
$\delta > 0$, introduce the shorthand $M = 2^{\pentropy_\optdomain(2
  \delta)}$ for the $2 \delta$ packing number, and form a collection
of points $\{\theta_1, \ldots, \theta_M \}$ that form a maximal $2
\delta$-packing of $\optdomain$.  Now consider any family of
conditional distributions $\{ \statprob( \cdot \, \mid \packval), \;
\packval \in [M] \}$ such that $\theta \big(\statprob( \cdot \, \mid
\, \packval ) \big) = \theta_\packval$.

Suppose that we sample an index $\packrv$ uniformly at random from
$[M]$, and then draw a sample $X \sim \statprob( \cdot \mid \packrv)$.
The associated testing problem is to determine the underlying
instantiation of the randomly chosen index.  Let $\msg = (\msg_1,
\ldots, \msg_T)$ denote the messages sent by the protocol $\protocol$,
and let $\what{\theta}(\msg)$ denote any estimator of $\theta$ based
on $\msg$.  Any such estimator defines a testing function via
\begin{equation*}
  \what{\packrv} \defn \argmin_{\packval \in \packset}
  \ltwos{\what{\theta}(\msg) - \theta_\packval}.
\end{equation*}
Since $\{ \theta_\packval\}_{\packval \in \packset}$ is a $2
\delta$-packing, we are guaranteed that $\ltwos{\what{\theta}(\msg) -
  \theta_\packval} \ge \delta$ whenever $\what{\packrv} \neq \packrv$,
whence
\begin{align}
  \max_{\packval \in \packset} \E\big[\ltwos{\what{\theta}(\msg) -
      \theta_\packval}^2\big] & \ge \sum_{\packval \in \packset}
  \P(\packrv = \packval) \E\big[\ltwos{\what{\theta}(\msg) -
      \theta_\packrv}^2 \mid \packrv = \packval\big] \nonumber \\ & \ge
  \sum_{\packval \in \packset} \delta^2 \P(\packrv = \packval)
  \P(\what{\packrv} \neq \packrv \mid \packrv = \packval) \, = \,
  \delta^2 \; \P(\what{\packrv} \neq \packrv).
  \label{eqn:centralized-testing-inequality}
\end{align}
It remains to lower bound the testing error $\P(\what{\packrv} \neq
\packrv)$.  Fano's inequality~\cite[Chapter 2]{Cover2006} yields
\begin{equation*}
  \P(\what{\packrv}\neq \packrv) \geq 1 - \frac{I(\packrv;\msg) +
    1}{\pentropy_\optdomain(2\delta)}.
\end{equation*}
Finally, the mutual information can be upper bounded as
\begin{align}
  \label{EqnMIBound}
  I(\packrv; \msg) & \stackrel{(i)}{\leq} \; H( \msg) \;
  \stackrel{(ii)}{\leq} \budget,
\end{align}
where inequality (i) is an immediate consequence of the definition of
mutual information, and inequality (ii) follows from Shannon's source
coding theorem~\cite{Cover2006}.  Combining
inequalities~\eqref{eqn:centralized-testing-inequality}
and~\eqref{EqnMIBound} yields
\begin{equation*}
  \minimaxinter(\theta, \pclass, \budget) \ge \delta^2 \Big \{ 1 -
  \frac{\budget + 1}{\pentropy_\optdomain(2\delta)} \Big \} \quad
  \mbox{for~any~} \delta > 0.
\end{equation*}
Because $1 - \frac{\budget +
  1}{\pentropy_\optdomain(2\delta)} \ge \half$ for any any choice of
$\delta$ such that $2 \delta \le \pentropy_\optdomain^{-1}(2 \budget +
2)$, setting $\delta = \half \pentropy_\optdomain^{-1}(2\budget
+ 2)$ yields the claim.

\subsection{A slight refinement}

We now describe a slight refinement of the classical reduction from
estimation to testing that underlies many of the remaining proofs.  It
is somewhat more general, since we no longer map the original
estimation problem to a strict test, but rather a test that allows
errors. We then leverage some variants of Fano's inequality developed
by a subset of the current authors~\cite{Duchi2013b}.

Defining $\packset = \{-1, +1\}^d$, consider an indexed family of
probability distributions \mbox{$\{\statprob(\cdot \mid
  \packval)\}_{\packval \in \packset} \subset \pclass$.}  Each member
of this family defines the parameter \mbox{$\theta_\packval \defeq
  \theta(\statprob(\cdot \mid \packval) )\in \Theta$.}  In particular,
suppose that we construct the distributions such that
\mbox{$\theta_\packval = \delta \packval$,} where $\delta > 0$ is a
fixed quantity that we control.  For any $\packval \neq \altpackval$,
we are then guaranteed that
\begin{align*}
  \ltwo{\theta_\packval - \theta_\altpackval} & = 2 \delta
  \sqrt{\dham(\packval,\altpackval)} \; \geq \; 2 \delta
\end{align*}
where $\dham(\packval,\altpackval)$ is the Hamming distance between
$\packval,\altpackval\in\packset$.  This lower bound shows that
$\{\theta_\packval \}_{\packval \in \packset}$ is a special type of
$2 \delta$-packing, in that the squared $\ell_2$-distance grows proportionally
to the Hamming distance between the indices $\packval$ and $\altpackval$.

Now suppose that we draw an index $\packrv$ from $\packset$ uniformly
at random, then drawing a sample $X$ from the distribution
$\statprob(\cdot \, \mid \, \packrv)$.  Fixing $t \ge 0$, the following
lemma~\cite{Duchi2013b} reduces the problem of estimating $\theta$ to
finding a point $\packval \in \packset$ within distance $t$ of the
random variable~$\packrv$.
\begin{lemma}
  \label{lemma:estimation-to-testing}
  Let $\packrv$ be uniformly sampled from $\packset$. For any estimator
  $\what{\theta}$ and any $t \ge 0 $, we have
  \begin{align*}
    \sup_{\statprob \in \pclass} \E[\ltwos{\what{\theta} -
        \theta(\statprob)}^2] & \geq \delta^2 \, (\floor{t} + 1) \;
    \inf_{\what{\packval}} \statprob \big(\dham(\what{\packval}, \packrv) > t
    \Big),
  \end{align*}
  where the infimum ranges over all testing functions $\what{\packval}$
  mapping the observations $X$ to $\packset$.
\end{lemma}

Setting $t = 0$, we recover the standard reduction from estimation to
testing as used in the proof of Proposition~\ref{PropGeneral}.  The
lemma allows for some additional flexibility in that it suffices to
show that, for some $t > 0$ to be chosen, it is difficult to identify
$\packrv$ within a Hamming radius of $t$. The following
variant~\cite{Duchi2013b} of Fano's inequality controls
this type of error probability:
\begin{lemma}
  \label{lemma:fun-fano}
  Let $\packrv \to \sample \to \what{\packrv}$ be a Markov chain,
  where $\packrv$ is uniform on $\packset$. For any $t \ge 0$, we have
  \begin{equation*}
    \P(\dham(\what{\packrv}, \packrv) > t)
    \ge 1 - \frac{I(\packrv; \sample) + \log 2}{
    \log\frac{|\packset|}{\hoodsize{t}}},
  \end{equation*}
  where $\hoodsize{t} \defeq \max\limits_{\packval\in
    \packset}|\{\altpackval\in \packset: \dham(\packval,\altpackval)\leq t
  \}|$ is the size of the largest $t$-neighborhood in $\packset$.
\end{lemma}

We thus have a clear avenue for obtaining lower bounds:
constructing a large packing set $\packset$ with (1) relatively small
$t$-neighborhoods, and (2) such that the mutual information $I(\packrv;
\sample)$ can be controlled.
Given this set-up, the remaining technical challenge is the
development of \emph{quantitative data processing inequalities}, which
allow us to characterize the effect of bit-constraints on the mutual
information $I(\packrv; \sample)$.  In general, these bounds are
significantly tighter than the trivial upper bound used in the proof
of Proposition~\ref{PropGeneral}.  Examples of such inequalities
in the sequel include
Lemmas~\ref{lemma:independent-message-contraction},
~\ref{lemma:gaussian-information-bounds},
and~\ref{lemma:interactive-gaussian-bounds}.


\subsection{Proof of Proposition~\ref{proposition:single-bernoulli-samples}}
\label{appendix:proof-theorem-single-bernoulli}

Given an index $\packval \in \packset$, suppose that each machine $i$
receives a $d$-dimensional sample $\sample^{(i)}$ with coordinates
independently sampled according to
\begin{equation*}
  P(\sample_j = \packval_j \mid \packval) = \frac{1 + \delta
    \packval_j}{2} ~~~ \mbox{and} ~~~ P(\sample_j = -\packval_j \mid
  \packval) = \frac{1 - \delta \packval_j}{2}.
\end{equation*}
Note that by construction, we have $\theta_\packval = \delta \packval
= \E_\packval[\sample]$, as well as
\begin{align}
  \label{EqnBerLRBound}
  \max_{x_j} \frac{P(x_j \mid \packval)}{P(x_j \mid \packval')} & \leq
  \frac{1 + \delta}{1 - \delta} = e^{\llratio} \qquad \mbox{where $\llratio
    \defeq \log \frac{1+ \delta}{1-\delta}$.}
\end{align}
Moreover, note that for any pair $(i,j)$, the sample
$\sample_j^{(i)}$, when conditioned on $\packrv_j$, is independent of
the variables \mbox{$\{\sample_{j'}^{(i)}:j'\neq
  j\}\cup\{\packrv_{j'}:j'\neq j\}$.}

Recalling that $\msg_i$ denotes the message sent by machine $i$,
consider the Markov chain \mbox{$\packrv \rightarrow \sample^{(i)}
  \rightarrow \msg_i$.}  By the usual data processing
inequality~\cite{Cover2006}, we have $I(\packrv; \msg_i) \leq
I(\sample^{(i)}; \msg_i)$.  The following result is a quantitative
form of this statement, showing how the likelihood ratio
bound~\eqref{EqnBerLRBound} causes a contraction in the mutual
information.
\begin{lemma}
  \label{lemma:independent-message-contraction}
  Under the preceding conditions, we have
  \begin{align*}
    I(\packrv; \msg_i) & \leq 2 (e^{2\llratio} - 1)^2 \; I(\sample^{(i)};
    \msg_i).
  \end{align*}
\end{lemma}
\noindent See Appendix~\ref{SecProofLemIndMesContract} for the proof
of this result. It is similar in spirit to recent results of
\citet[Theorems 1--3]{Duchi2013}, who establish quantitative data
processing inequalities in the context of privacy-preserving data
analysis. Our proof, however, is different, as we have the Markov
chain $\packrv \to \sample \to \msg$, and instead of a likelihood
ratio bound on the channel $\sample \to \msg$ as
in the paper~\cite{Duchi2013}, we place a likelihood
ratio bound on $\packrv \to \sample$.

Next we require a certain tensorization property of the mutual
information, valid in the case of independent protocols:
\begin{lemma}
  \label{LemTensorize}
  When $\msg_i$ is a function only of $\sample^{(i)}$, then
  \begin{align*}
    I(\packrv; \msg_{1:\nummac}) &  \leq  \sum_{i=1}^\nummac I(\packrv; \msg_i).
  \end{align*}
\end{lemma}
\noindent See Appendix~\ref{AppTensorize} for a proof of this claim.\\

\noindent We can now complete the proof of the proposition.  Using
Lemma~\ref{lemma:independent-message-contraction}, we have
\begin{align*}
  I(\packrv; \msg_i) & \leq 2 \bigg(e^{2 \log \frac{1 + \delta}{1 -
      \delta}} - 1\bigg)^2 I(\sample^{(i)}; \msg_i) \; = \; 2
  \big(\frac{(1 + \delta)^2}{(1 - \delta)^2} - 1\big)^2 \; \leq \; 80
  \delta^2 I(\sample^{(i)}; \msg_i),
\end{align*}
valid for $\delta \in [0, 1/5]$.  Applying Lemma~\ref{LemTensorize}
yields
\begin{align*}
  I(\packrv; \msg_{1:\nummac}) & \leq \sum_{i = 1}^\nummac I(\packrv;
  \msg_i) \le 80 \delta^2 \sum_{i = 1}^\nummac I(\msg_i;
  \sample^{(i)}).
\end{align*}
The remainder of the proof is broken into two cases, namely $d \geq 10$
and $d < 10$.

\paragraph{Case $d \geq 10$:}
By the definition of mutual information, we have
\begin{align*}
  I(\msg_i; \sample^{(i)}) & \leq \min\{H(\msg_i), H(\sample^{(i)})\} \; \leq \;
  \min \{ \budget_i, d \},
\end{align*}
where the final step follows since $H(\sample^{(i)}) \leq d$ and
$H(\msg_i) \leq \budget_i$, the latter inequality following from
Shannon's source coding theorem~\cite{Cover2006}.  Putting together the pieces,
we have
\begin{align*}
  I(\packrv; \msg_{1:\nummac}) & \leq 80 \delta^2 \sum_{i=1}^\nummac
  \min \{\budget_i, d \}.
\end{align*}
Combining this upper bound on mutual information with
Lemmas~\ref{lemma:estimation-to-testing} and~\ref{lemma:fun-fano}
yields the lower bound
\begin{align*}
  \minimaxind(\theta, \pclass, \budget_{1:\nummac}) & \geq \delta^2
  (\floor{d/6} + 1) \bigg( 1 - \frac{80 \delta^2 \sum_{i=1}^\nummac
    \min\{\budget_i, d\} + \log 2}{d/6}\bigg).
\end{align*}
The choice $\delta^2 = \min\{1/25, d / 960 \sum_{i=1}^\nummac
\min\{\budget_i, d\}\}$ guarantees that the expression inside
parentheses in the previous display is lower bounded by $2/25$, which
completes the proof for $d \ge 10$.

\paragraph{Case $d < 10$:}  In this case, we make use of Le Cam's
method instead of Fano's method.  More precisely, by reducing to a
smaller dimensional problem, we may assume without loss of generality
that $d = 1$, and we set $\packset = \{-1, 1\}$.  Letting $\packrv$ be
uniformly distributed on $\packset$, the Bayes error for binary
hypothesis testing is (e.g.~\cite[Chapter 2]{Yu1997,Tsybakov2009})
\begin{align*}
  \inf_{\what{\packval}} \P(\what{\packval} \neq \packrv) = \half -
  \half \tvnorm{\statprob_1 - \statprob_{-1}}.
\end{align*}
As $\theta_\packval = \delta\packval$ by construction, the
reduction from estimation to testing in
Lemma~\ref{lemma:estimation-to-testing} implies
\begin{align*}
  \inf_{\what{\theta}} \max_{\statprob \in \{\statprob_1,
    \statprob_{-1}\}} \E[\ltwos{\what{\theta} - \theta(\statprob)}^2]
  \ge \delta^2 \bigg(\half - \half \tvnorm{\statprob_1 -
    \statprob_{-1}}\bigg).
\end{align*}
Finally, as we show in Appendix~\ref{AppSmall}, we have
the following consequence of Pinsker's inequality:
\begin{equation}
  \label{eqn:tv-norm-to-info}
  \tvnorm{\statprob_\msg(\cdot \mid \packrv = \packval) -
    \statprob_\msg(\cdot \mid \packrv = \altpackval)}^2 \le 2 I(\msg;
  \packrv).
\end{equation}
Thus
\begin{align}
  \label{eqn:bernoulli-low-dimensional-reduction}
  \minimaxind(\theta, \pclass, \budget_{1:\nummac}) \ge \delta^2
  \bigg(\half - \half \sqrt{2 I(\packrv; \msg_{1:\nummac})}\bigg).
\end{align}
Arguing as in the previous case $(d \ge 10)$, we have the upper bound
$I(\sample^{(i)}; \msg_i) \le \min\{\budget_i, 1\}$, and hence
\begin{equation*}
  \minimaxind(\theta, \pclass, \budget_{1:\nummac}) \ge \delta^2
  \bigg[\half - 7 \bigg(\delta^2 \sum_{i=1}^\nummac \min\{\budget_i,
    1\} \bigg)^\half \bigg].
\end{equation*}
Setting $\delta^2 = \min\big\{\frac{1}{25}, \frac{1}{400 \sum_{i=1}^\nummac
  \min\{\budget_i, 1\}}\big\}$ completes the proof.


\subsection{Proof of Theorem~\ref{theorem:gaussian-mean-communication}}
\label{appendix:proof-gaussian-mean-communication}

This proof follows a similar outline to that of
Proposition~\ref{proposition:single-bernoulli-samples}.
We assume that the sample $\sample^{(i)}$ at machine~$i$ contains $\numobs_i$
independent observations from the multivariate normal distribution,
and we will use the fact that $\numobs_i\equiv \numobs$ at
the end of the proof, demonstrating that the proof technique
is sufficiently general to allow for different sized subsets
in each machine.
We represent the
$i$th as a $d \times \numobs_i$ matrix
$\sample^{(i)} \in \R^{d \times \numobs_i}$.  We use $\sample^{(i,k)}$
and $\sample^{(i)}_{j}$ to denote, respectively, the $k$th column
and $j$th row of this matrix.  Throughout this argument, we assume
that $\nummac \ge 5$; otherwise, Proposition~\ref{PropGeneral}
provides a stronger result.

As in the previous section, we consider a testing problem in which the
index $\packrv \in \{-1, +1\}^d$ is drawn uniformly at random.
Our first step is to provide a quantitative data processing inequality analogous
to Lemma~\ref{lemma:independent-message-contraction}, but which
applies in somewhat more general settings. To that end, we abstract a bit
from our current setting, and consider a model such that
for
any $(i,j)$, we assume that given
$\packrv_j$, the $j$th row
row $\sample_j^{(i)}$ is conditionally independent of all other rows
$\{\sample_{j'}^{(i)}:j'\neq j\}$ and all other packing indices
$\{\packrv_{j'}:j'\neq j\}$.  In addition, letting
$\statprob_{\sample_j}$ denote the probability measure of
$\sample_j^{(i)}$, we assume that there exist measurable sets $\measset_j
\subset \range(\sample_j^{(i)})$ such that
\begin{align*}
 \sup_{S \in \sigma(\llset_j)} \frac{\statprob_{\sample_j}(S \mid
   \packrv = \packval)}{ \statprob_{\sample_j}(S \mid \packrv =
   \altpackval)} \le \exp(\llratio),
\end{align*}
Let $E_j$ be a $\{0,1\}$-valued indicator variable for the event
$\sample_j^{(i)} \in \llset_j$ (i.e.\ $E_j = 1$ iff $\sample_j^{(i)} \in
\llset_j$, and we leave the indexing on $i$ implicit). We have the
following bound:
\begin{lemma}
  \label{lemma:independent-message-unbounded}
  Under the conditions stated in the preceding paragraph, we have
  \begin{align*}
    I(\packrv; \msg_i) \le 2\big(e^{4 \llratio} - 1\big)^2
    I(\sample^{(i)}; \msg_i) + \sum_{j=1}^d H(E_j) + \sum_{j = 1}^d P(E_j
    = 0).
  \end{align*}
\end{lemma}
\noindent See Appendix~\ref{SecProofLemIndMesUnb} for the proof of
this claim. \\

Our next step is to bound the terms involving the indicator variables
$E_j$.  Fixing some $\delta > 0$, for each $\packval \in \{-1, 1\}^d$
define $\theta_\packval = \delta \packval$, and conditional on
$\packrv = \packval \in \{-1, 1\}^d$, let $\sample^{(i, k)}$, $k = 1,
\ldots, \numobs_i$, be drawn i.i.d.\ from a $\normal(\theta_\packval,
\sigma^2 I_{d \times d})$ distribution.  The following lemma applies
to any pair of non-negative numbers $(a, \delta)$ such that
\begin{align}
  \label{EqnHourglass}
  \max_{i \in [\nummac]} \frac{\sqrt{\numobs_i} a \delta}{\sigma^2}
  \le \frac{1}{4} \quad \mbox{and} \quad a \ge \delta \max_{i \in [\nummac]}
  \sqrt{\numobs_i}.
\end{align}
It also involves the binary entropy function $h_2(p) \defeq -p
\log_2(p) - (1-p) \log_2(1-p)$.
\begin{lemma}
  \label{lemma:gaussian-information-bounds}
  For any pair $(a, \delta)$ satisfying condition~\eqref{EqnHourglass},
  we have
  \begin{subequations}
    \begin{align}
      \label{eqn:markov-bound-application}
      I(\packrv; \msg_i) & \leq \frac{d\numobs_i\delta^2}{\sigma^2}, ~~~~~~
      \mbox{and} \\
      \label{eqn:transition-bound-application}
      I(\packrv; \msg_i) & \leq 128 \frac{\delta^2 a^2}{\sigma^4} \numobs_i
      H(\msg_i) + d \, h_2(p_i^*) +
      d \, p_i^*,
    \end{align}
  \end{subequations}
  where $p_i^* \defeq \min\left\{2 \exp\big(-\frac{(a - \sqrt{\numobs_i}
    \delta)^2}{ 2 \sigma^2} \big), \half \right\}$.
\end{lemma}

With the bounds~\eqref{eqn:markov-bound-application}
and~\eqref{eqn:transition-bound-application} on the mutual information
$I(\msg_i; \packrv)$, we may now divide our proof into two cases: when
$d < 10$ and $d \geq 10$.


\paragraph{Case $d \geq 10$:}
In this case, we require an additional auxiliary result, which we prove
via Lemma~\ref{lemma:gaussian-information-bounds}.
(See Appendix~\ref{AppLemIceAxe} for the proof of this claim.)
\begin{lemma}
  \label{LemIceAxe}
  For all $\delta \in \big[0, \frac{\sigma}{16} (\log \nummac \,
      \max_i \numobs_i )^{-\half} \big]$, we
  have
  \begin{equation}
    \label{eqn:theorem-1-mutual-information-final-bound}
    \sum_{i=1}^\nummac I(\packrv; \msg_i)
    \le \delta^2 \sum_{i=1}^\nummac \frac{\numobs_i}{\sigma^2}
    \min\big\{128 \cdot 16 \log \nummac \cdot H(\msg_i), d
    \big\}
    + d \bigg(\frac{2}{49} + 2 \cdot 10^{-5}\bigg).
  \end{equation}
\end{lemma}
\noindent 
Combining the upper bound~\eqref{eqn:theorem-1-mutual-information-final-bound}
on the mutual information with the minimax lower
bounds in Lemmas~\ref{lemma:estimation-to-testing} and~\ref{lemma:fun-fano},
and noting that $6 (2/49 + 2 \cdot 10^{-5}) + 6 \log 2 / d \le 2/3$ when $d
\geq 10$ yields the following minimax bound:
\begin{align}
  \label{eqn:non-interactive-min-bound}
  \minimaxind(\theta, \mc{P}, \budget_{1:\nummac}) \ge \delta^2
  \big(\floor{d/6} + 1\big) \bigg(\frac{1}{3} - \frac{6 \delta^2
    \sum_{i=1}^\nummac \numobs_i \min \{128 \cdot 16 \log \nummac \cdot
    H(\msg_i), d\}}{d \sigma^2 }\bigg).
\end{align}

Using this result, we now complete the proof of the theorem By
Shannon's source coding theorem, we have $H(\msg_i) \le \budget_i$,
whence the minimax bound~\eqref{eqn:non-interactive-min-bound} becomes
\begin{equation*}
  \delta^2 \big(\floor{d/6} + 1\big) \bigg(\frac{1}{3} - \frac{6
    \delta^2 \sum_{i=1}^\nummac \numobs_i \min \{128 \cdot 16 \budget_i
    \log \nummac , d\}}{d \sigma^2 }\bigg).
\end{equation*}
In particular, if we choose
\begin{equation}
  \label{eqn:theorem-1-choice-of-delta}
  \delta^2 = \min\bigg\{1, \frac{\sigma^2}{16^2 \max_i \numobs_i \log
    \nummac}, \frac{d \sigma^2}{36 \sum_{i=1}^\nummac \numobs_i \min
    \{128 \cdot 16 \budget_i \log \nummac, d\}}\bigg\},
\end{equation}
we obtain
\begin{equation*}
  \frac{1}{3} - \delta^2 \frac{6 \sum_{i=1}^\nummac \numobs_i
    \min\{128 \cdot 16 \budget_i \log \nummac, d\}}{d \sigma^2} \ge
  \frac{1}{6},
\end{equation*}
which yields the minimax lower bound
\begin{equation*}
  \minimaxind(\theta, \mc{P}, \budget_{1:\nummac}) \ge
  \frac{1}{6}\big(\floor{d/6} + 1\big) \min\bigg\{1,
  \frac{\sigma^2}{16^2 \max_i \numobs_i \log \nummac}, \frac{d
    \sigma^2}{36 \sum_{i=1}^\nummac \numobs_i \min \{128 \cdot 16
    \budget_i \log \nummac, d\}}\bigg\}.
\end{equation*}
To obtain inequality~\eqref{eqn:gaussian-mean-communication}, we
simplify by assuming that $\numobs_i \equiv \numobs$ for all $i$ and
perform simple algebraic manipulations, noting that the minimax lower
bound $d \sigma^2 / (\numobs \nummac)$ holds independently of any
communication budget.

\paragraph{Case $d <  10$:}
As in the proof of Proposition~\ref{proposition:single-bernoulli-samples}, we
cover this case by reducing to dimension $d = 1$ and applying Le Cam's
method, in particular via the lower
bound~\eqref{eqn:bernoulli-low-dimensional-reduction}.
Substituting in the $\delta^2$ assignment~\eqref{eqn:theorem-1-choice-of-delta}
and the relation $H(\msg_i) \le \budget_i$ into Lemmas~\ref{LemTensorize}
and~\ref{LemIceAxe}, we find that

\begin{align*}
  I(\packrv; \msg_{1:\nummac})
  \le \sum_{i=1}^\nummac I(\packrv; \msg_i)
  \leq \frac{1}{36} + \frac{2}{49} + 2 \cdot 10^{-5} < \frac{1}{8}.
\end{align*}
Applying Le Cam's method to this upper bound implies the lower bound
 $\minimaxind(\theta, \mc{P}, \budget_{1:\nummac}) \ge \delta^2/4$,
which completes the proof.


\subsection{Proof of Proposition~\ref{corollary:uniform-upper-bound}}
\label{sec:proof-uniform}

Proposition~\ref{corollary:uniform-upper-bound} involves both a lower
and upper bound.  We prove the upper bound by exhibiting a specific
interactive protocol $\protocol^*$, and the lower bound via an
application of Proposition~\ref{PropGeneral}.

\paragraph{Proof of lower bound:}

Applying Proposition~\ref{PropGeneral} requires a lower bound on the
packing entropy of $\optdomain = [-1, 1]^d$.  By a standard volume
argument~\cite{Ball1997}, the $2 \delta$-packing entropy has lower bound
\begin{align*}
  \pentropy_\optdomain(2\delta) & \geq \log_2
  \frac{\mbox{Volume}(\Theta)}{ \mbox{Volume} (\{x\in \R^d: \ltwos{x}
    \leq 2 \delta \} )} \geq d \log \Big(\frac{1}{2\delta} \Big).
\end{align*}
Inverting the relation $\budget = \pentropy_\optdomain(\delta) =
\pentropy_\optdomain(1 / (\nummac \numobs))$ yields the lower bound.

\paragraph{Proof of upper bound:}
Consider the following communication protocol $\protocol^* \in
\pinter$:
\begin{enumerate}[(i)]
\item Each machine $i\in[\nummac]$ computes its local minimum $a_j^{(i)} =
  \min\{X_j^{(i,k)}: k\in[\numobs]\}$ for each coordinate $j \in [d]$.
\item Machine $1$ broadcasts the vector $a^{(1)}$, where each of its
  components is quantized to accuracy $(\nummac \numobs)^{-2}$ in $[-2, 2]$,
  rounding down, using $2 d \log_2(2 \nummac \numobs)$ bits. Upon receiving
  the broadcast, all machines initialize global minimum variables $s_j
  \leftarrow a^{(1)}_j$ for $j = 1, \ldots, d$.
\item In the order $i=2,3,\dots,\nummac$, machine $i$ performs
  the following operations:
  \begin{enumerate}[(i)]
  \item Find all indices $j$ such that $a^{(i)}_j < s_j$, calling this set
    $J_i$. For each index $j \in J_i$, machine $i$ updates $s_j \leftarrow
    a^{(i)}_j$, and then broadcasts the list of indices $J_i$ (which
    requires $|J_i| \ceil{\log_2 d}$ bits) and the associated values
    $s_j$, using a total of $|J_i| \ceil{\log_2 d} + 2 |J_i| \log(2
    \nummac \numobs)$ bits.
  \item All other machines update their local vectors $s$ after receiving
    machine $i$'s update.
  \end{enumerate}
\item One machine outputs $\what{\theta} = s+1$.
\end{enumerate}

Using the protocol $\protocol^*$ above, it is clear that for each
$j \in [d]$ we have computed the global minimum
\begin{equation*}
  s_j^* = \min\big\{\sample^{(i,k)}_j \mid i \in [\nummac], k \in [\numobs]
  \big\}
\end{equation*}
to within accuracy $1 / (\nummac \numobs)^2$ (because of
quantization).  As a consequence, classical convergence analyses
(e.g.~\cite{Lehmann1998}) yield that the estimator $\what{\theta} = s
+ 1$ achieves the minimax optimal convergence rate
$\E[\ltwos{\what{\theta}-\theta}^2] \le c \frac{d}{(\nummac
  \numobs)^2}$, where $c > 0$ is a numerical constant.

It remains to understand the communication complexity of the protocol
$\protocol^*$. To do so, we study steps 2 and 3.
In Step 2, machine $1$ sends a $2 d \log_2(2 \nummac\numobs)$-bit message
as $\msg_1$. In Step 3, machine $i$ sends
$|J_i| (\ceil{\log_2 d} + 2 \log_2(2 m n))$ bits, that is, at most
\begin{equation*}
  \sum_{j=1}^d \indic{a^{(i)}_j < \min\{a^{(1)}_j, \ldots, a^{(i-1)}_j\}}
  (\ceil{\log_2 d} + 2 \log_2(2 \nummac \numobs))
\end{equation*}
bits, as no message is sent for index $j$ if $a^{(i)}_j \ge
\min\{a^{(1)}_j, \ldots, a^{(i-1)}_j\}$.
By inspection, this event
happens with probability bounded by $1/i$, so we find that
the expected length of message $\msg_i$ is
\begin{equation*}
  \E[\length_i] \leq \frac{d(\ceil{\log_2 d} + 2\log_2(2\nummac\numobs))}{i}.
\end{equation*}
Putting all pieces together, we obtain that
\begin{align*}
  \E[\length] = \sum_{i=1}^\nummac \E[\length_i]
  & \leq 2d\log(2\nummac\numobs)
  + \sum_{i=2}^\nummac
  \frac{d(\ceil{\log_2 d} + 2\log_2(2\nummac\numobs))}{i} \\
  & \le
  d\big[2 \log_2(2 \nummac \numobs)
  + \log(\nummac)(\ceil{\log d} + 2 \log_2(2 \nummac\numobs))\big].
\end{align*}


\subsection{Proof of Theorem~\ref{theorem:interactive-gaussian}}
\label{sec:proof-interactive-gaussian}

As in the proof of Theorem~\ref{theorem:gaussian-mean-communication},
we choose $\packrv \in \{-1,1\}^d$ uniformly at random, and for some
$\delta > 0$ to be chosen, we define the parameter vector $\theta
\defeq \delta \packrv$.  Suppose that machine $i$ draws a sample
$\sample^{(i)} \in \R^{d \times \numobs}$ of size $\numobs$
i.i.d.\ according to a $\normal (\theta, \sigma^2 I_{d\times d})$
distribution.  We denote the full sample---across all machines---along
dimension $j$ by $\sample_j$. In addition, for each $j \in [d]$, we
let $\packrv_{\setminus j}$ denote the coordinates of $\packrv \in
\{-1, 1\}^d$ except the $j$th coordinate.

Although the local samples are independent, since we now allow for
interactive protocols, the messages can be dependent: the sequence of
random variables $\msg = (\msg_1,\dots,\msg_\nummsg)$ is generated in
such a way that the distribution of $\msg_t$ is
$(X^{(i_t)},\msg_{1:t-1})$-measurable, where
$i_t\in\{1,\dots,\nummac\}$ is the machine index upon which $\msg_t$
is based (i.e.\ the machine sending message $\msg_t$). We assume
without loss of generality that the sequence $\{i_1, i_2, \ldots,\}$
is fixed in advance: if the choice of index $i_t$ is not fixed but
chosen based on $\msg_{1:t-1}$ and $\sample$, we simply say there
exists a default value (say no communication or $\msg_t = \perp$) that
indicates ``nothing'' and has no associated bit cost.

To prove our result, we require an analogue of
Lemma~\ref{lemma:independent-message-unbounded} (cf.\ the proof of
Theorem~\ref{theorem:gaussian-mean-communication}).  Assuming temporarily that
$d = 1$, we prove our analogue for one-dimensional interactive protocols, and
in the sequel, we show how it is possible to we reduce multi-dimension
problems to this statement.  As in the proof of
Theorem~\ref{theorem:gaussian-mean-communication}, we abstract a bit from our
specific setting, instead assuming a likelihood ratio constraint, and provide
a data processing inequality for our setting.  Let $\packrv$ be a Bernoulli
variable uniformly distributed on $\{-1, 1\}$, and let
$\statprob_{\sample^{(i)}}$ denote the probability measure of the $i$th sample
$\sample^{(i)}\in \R^{\numobs}$.  Suppose there is a (measurable) set $\llset$
such that for any $\packval, \packval' \in \{-1, 1\}$, we have
\begin{align}
  \label{eqn:likelihood-ratio-subset-bound}
  \sup_{S \in \sigma(\llset)} \frac{\statprob_{\sample^{(i)}}(S \mid
    \packval)}{ \statprob_{\sample^{(i)}}(S \mid \packval')} & \leq
  e^{\llratio}.
\end{align}
Finally, let $E$ be a $\{0,1\}$-valued indicator variable for the
event $\cap_{i=1}^\nummac \{ \sample^{(i)}\in \llset \}$.

\begin{lemma}
  \label{lemma:information-contraction-in-subset}
  Under the previously stated conditions, we have
  \begin{align*}
    I(\packrv; \msg) & \leq 2\big(e^{4 \llratio} - 1\big)^2 I(\sample;
    \msg) + H(E) + P(E = 0).
  \end{align*}
\end{lemma}
\noindent
\noindent See Appendix~\ref{sec:proof-contraction-subset} for the
proof.

Using this lemma as a building block, we turn to the case that
$\sample^{(i)}$ is $d$-dimensional.
Making an explicit choice of
the set $\llset$, we obtain the following concrete bound on the mutual
information. The lemma applies to any pair $(a, \delta)$ of non-negative
reals such that
\begin{align*}
  \frac{\sqrt{\numobs} a \delta}{\sigma^2} \le \frac{1}{4} \quad
  \mbox{and} \quad a \ge \delta \sqrt{\numobs},
\end{align*}
and, as in Lemma~\ref{lemma:gaussian-information-bounds}, involves the binary
entropy function $h_2(p) \defeq -p \log(p) - (1 - p) \log(1 - p)$.

\begin{lemma}
  \label{lemma:interactive-gaussian-bounds}
  Under the preceding conditions, we have
  \begin{align*}
    I(\packrv_j; \msg \mid \packrv_{\setminus j}) & \le 128
    \frac{\delta^2\numobs a^2}{\sigma^4} I(\sample_j; \msg \mid
    \packrv_{\setminus j}) + \nummac h_2(p^*) + \nummac
    p^*
  \end{align*}
  where $p^* \defeq \min \left\{ 2 \exp\big(-\frac{(a - \sqrt{\numobs}
    \delta)^2}{ 2 \sigma^2} \big), \half \right\}$.
\end{lemma}
\noindent We prove the lemma in
Section~\ref{sec:proof-interactive-gaussian-information}.

To apply Lemma~\ref{lemma:interactive-gaussian-bounds}, we require
two further intermediate bounds on mutual information terms.  By the chain
rule for mutual information~\cite{Cover2006}, we have
\begin{align*}
  I(\packrv; \msg) = \sum_{j = 1}^d I(\packrv_j; \msg \mid
  \packrv_{1:j-1}) \; & = \; \sum_{j=1}^d [H(\packrv_j \mid
    \packrv_{1:j-1}) - H(\packrv_j \mid \msg, \packrv_{1:j-1})] \\
  & \stackrel{(i)}{=} \sum_{j=1}^d [H(\packrv_j \mid \packrv_{\setminus
      j}) - H(\packrv_j \mid \msg, \packrv_{1:j-1})],
\end{align*}
where equality (i) follows since the variable $\packrv_j$ is
independent of $\packrv_{\setminus j}$.  Since conditioning can only
reduce entropy, we have $H(\packrv_j \mid \msg, \packrv_{1:j-1}) \geq
H(\packrv_j \mid \msg, \packrv_{\setminus j})$, and hence
\begin{align}
  \label{eqn:left-condition-to-uncondition-for-interactive-gaussian}
  I(\packrv; \msg) & \leq \sum_{j=1}^d [H(\packrv_j \mid
    \packrv_{\setminus j}) - H(\packrv_j \mid \msg, \packrv_{\setminus
      j})] \; = \; \sum_{j=1}^d I(\packrv_j; \msg \mid
  \packrv_{\setminus j}).
\end{align}

Turning to our second intermediate bound, by the definition
of the conditional mutual information, we have
\begin{align*}
  \sum_{j=1}^d I(\sample_j; \msg \mid \packrv_{-j}) \: = \:
  \sum_{j=1}^d [ H(\sample_j \mid \packrv_{\setminus j}) - H(\sample_j
    \mid \msg, \packrv_{\setminus j})] & \stackrel{(i)}{=} H(\sample)
  - \sum_{j=1}^d H(\sample_j \mid \msg, \packrv_{\setminus j}) \\
  & \stackrel{(ii)}{\le} H(\sample) - \sum_{j=1}^d H(\sample_j \mid
  \msg, \packrv) \\
  & \stackrel{(iii)}{\leq} H(\sample) - H(\sample \mid \msg, \packrv) =
  I(\sample; \msg, \packrv),
\end{align*}
where equality (i) follows by the independence of $\sample_j$ and
$\packrv_{\setminus j}$, inequality (ii) because conditioning reduces
entropy, and inequality (iii) because $H(\sample \mid \msg, \packrv)
\le \sum_j H(\sample_j \mid \msg, \packrv)$.  Noting that $I(\sample;
\packrv, \msg) \le H(\packrv, \msg) \le H(\msg) + d$, we conclude that
\begin{equation}
  \label{eqn:right-condition-to-uncondition-for-interactive-gaussian}
  \sum_{j=1}^d I(\sample_j; \msg \mid \packrv_{\setminus j})
  \le I(\sample; \packrv, \msg) \le H(\msg) + d.
\end{equation}

We can now complete the proof of the theorem.  Combining
inequalities~\eqref{eqn:left-condition-to-uncondition-for-interactive-gaussian}
and~\eqref{eqn:right-condition-to-uncondition-for-interactive-gaussian}
with Lemma~\ref{lemma:interactive-gaussian-bounds} yields
\begin{align}
  \label{eqn:global-mutual-information-bound-for-interative-gaussian}
  I(\packrv; \msg) & \le 128 \frac{\delta^2\numobs a^2}{\sigma^4}
  (H(\msg) + d) + \nummac \, d \: h_2(p^*) + m d p^*,
\end{align}
where we recall that $p^* = \{ 2 \exp\big(-\frac{(a - \sqrt{\numobs}
  \delta)^2}{ 2 \sigma^2} \big), \half \}$.

Inequality~\eqref{eqn:global-mutual-information-bound-for-interative-gaussian}
is the analog of inequality~\eqref{eqn:transition-bound-application}
in the proof of Theorem~\ref{theorem:gaussian-mean-communication};
accordingly, we may follow the same steps to complete the proof.
The case $d < 10$ is entirely analogous; the case $d \geq 10$
involves a few minor differences that we describe here. \\

Setting $a = 4 \sigma \sqrt{\log \nummac}$, choosing some $\delta$ in
the interval $[0, \frac{\sigma}{16 \sqrt{\numobs \log \nummac}}]$, and
then applying the
bound~\eqref{eqn:global-mutual-information-bound-for-interative-gaussian},
we find that
\begin{equation*}
  I(\packrv; \msg) \le \delta^2 \frac{128 \cdot 16 \numobs \log
    \nummac}{ \sigma^2} (H(\msg) + d) + d\left(\frac{2}{49} + 2 \cdot
  10^{-5}\right).
\end{equation*}
Combining this upper bound on the mutual information with
Lemmas~\ref{lemma:estimation-to-testing} and~\ref{lemma:fun-fano}, we find that
\begin{align*}
  \minimaxinter(\theta, \pclass, \budget) & \geq \delta^2 (\floor{d/6}
  + 1) \bigg[\frac{1}{3} - (128 \cdot 16 \cdot 6) \delta^2
    \frac{(H(\msg) + d) \numobs \log \nummac}{d \sigma^2}\bigg] \\
  & \geq \delta^2 (\floor{d/6} + 1) \bigg[\frac{1}{3} - (128 \cdot 16
    \cdot 6) \delta^2 \frac{(\budget + d) \numobs \log \nummac}{d
      \sigma^2}\bigg],
\end{align*}
where the second step follows since $H(\msg) \leq \budget$, by the
source coding theorem~\cite{Cover2006}. Setting
\begin{equation*}
  \delta^2 = \min\left\{1, \frac{\sigma^2}{256 \numobs \log \nummac },
  \frac{d \sigma^2}{2048 \cdot 36 \cdot \numobs (\budget + d) \log
    \nummac}\right\}
  = \min\left\{1, \frac{d \sigma^2}{2048 \cdot 36
    \cdot \numobs (\budget + d) \log \nummac}\right\},
\end{equation*}
we obtain
\begin{align*}
  \minimaxinter(\theta, \pclass, \budget) & \geq \delta^2
  \frac{\floor{d/6} + 1}{6} \; = \; \min\left\{1, \frac{d \sigma^2}{2048
    \cdot 36 \cdot \numobs (\budget + d) \log \nummac}\right\} \;
  \frac{\floor{d/6} + 1}{6}
\end{align*}
Noting that $\minimaxinter(\theta, \pclass, \budget)
\ge \minimaxinter(\theta, \pclass, \infty)
\gtrsim \frac{\sigma^2 d}{\numobs \nummac}$ completes the proof.


\section{Discussion}
\label{SecDiscuss}

In this paper, we have established lower bounds on the amount of
communication required for several statistical estimation problems.
Our lower bounds are information-theoretic in nature, based on
variants of Fano's and Le Cam's methods.  In particular, they rely on
novel types of quantitative data processing inequalities that
characterize the effect of bit constraints on the mutual information
between parameters and messages.
Several open questions remain. Our arguments are somewhat complex, and our
upper and lower bounds differ by logarithmic factors. It would be interesting
to understand which of our bounds can be sharpened; tightening the upper
bounds would lead to interesting new distributed inference protocols, while
improving the lower bounds could require new technical insights. We believe it
will also be interesting to explore the application and extension of
our results and techniques to other---perhaps more complex---problems in
statistical estimation.

\subsection*{Acknowledgements}

This work was supported in part by the U.S.\ Army Research Laboratory,
U.S.\ Army Research Office grant W911NF-11-1-0391, Office of Naval
Research MURI grant N00014-11-1-0688, and National Science Foundation
grant CIF-31712-23800.  In addition, JCD was supported in part by a
Facebook Graduate Fellowship.


\appendix

\section{Contractions in total variation distance}
\label{sec:total-variation-contraction}

As noted in the main body of the paper, our results rely on certain
quantitative data processing inequalities.  They are inspired by results on
information contraction under privacy constraints developed by a subset of the
current authors (\citet{Duchi2013}). In this appendix, we present a technical
result---a contraction in total variation distance---that underlies many of
our proofs of the data processing inequalities
(Lemmas~\ref{lemma:independent-message-contraction},
\ref{lemma:independent-message-unbounded},
and \ref{lemma:information-contraction-in-subset}).

Consider a random vector $(A, B, C, D)$ with
joint distribution $P_{A,B,C,D}$, where $A$, $C$ and $D$ take on
discrete values.  Denoting the conditional distribution of $A$ given
$B$ by $P_{A \mid B}$, suppose that $(A, B, C, D)$ respect the
conditional independence properties defined by the directed graphical
model in Figure~\ref{fig:information-chaining}.  In analytical terms,
we have
\begin{equation}
  \label{eqn:joint-distribution-decomposition}
  P_{A,B,C,D} = P_A P_{B\mid A} P_{C \mid A, B} P_{D \mid B, C}.
\end{equation}
In addition, we assume that there exist functions $\fpotl_1 : \mc{A}
\times \sigma(\mc{C}) \rightarrow \R_+$ and $\fpotl_2 : \mc{B} \times
\sigma(\mc{C}) \rightarrow \R_+$ such that
\begin{equation}
  \label{eqn:message-distribution-factorization}
  P_{C}(S \mid A, B) = \fpotl_1(A, S) \fpotl_2(B, S)
\end{equation}
for any (measureable) set $S$ in the range $\mc{C}$ of $C$.  Since $C$
is assumed discrete, we abuse notation and write $P(C = c \mid A, B) =
\fpotl_1(A, c) \fpotl_2(B, c)$.  Lastly, suppose that
\begin{equation}
  \label{eqn:likelihood-constraint}
  \sup_{S \in \sigma(\mc{B})} \frac{P_B(S \mid A = a)}{P_B(S \mid A =
    a')} \le \exp(\llratio) \qquad \mbox{for all $a, a' \in \mc{A}$.}
\end{equation}
The following lemma applies to the absolute difference
\begin{align*}
  \Delta(a, C, D) & \defeq \big|P(A = a \mid C, D) - P(A = a \mid
  C)\big|.
\end{align*}

\begin{figure}
  \begin{center}
    \begin{overpic}[width=.4\columnwidth]
      {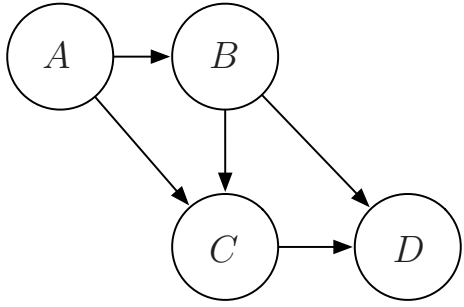} \put(11.5,51){\Large $A$}
      \put(45,51){\Large $B$} \put(45,12){\Large $C$}
      \put(82,12){\Large $D$}
    \end{overpic}
    \caption{\label{fig:information-chaining} Graphical model
      for Lemma~\ref{lemma:information-chaining}}
  \end{center}
\end{figure}
\begin{lemma}
  \label{lemma:information-chaining}
  Under conditions~\eqref{eqn:joint-distribution-decomposition},
  \eqref{eqn:message-distribution-factorization},
  and~\eqref{eqn:likelihood-constraint}, we have
  \begin{align*}
    \Delta(a, C, D) & \leq 2 \big(e^{2 \llratio} - 1\big) \min\big\{P(A =
    a \mid C), P(A = a \mid C, D)\big\} \tvnorm{P_B(\cdot \mid C, D) -
      P_B(\cdot \mid C)}.
  \end{align*}
\end{lemma}
\noindent

\begin{proof}
  By assumption, $A$ is independent of $D$ given $\{B,C\}$. Thus we may
  write
  \begin{align*}
    \Delta(a, C, D) & = \Big| \int P(A=a \mid B = b,C)\big(dP_B(b \mid C ,
    D) - dP_B(b \mid C) \big) \Big|.
  \end{align*}
  Combining this equation with the relation $\int_{\mc{B}} P(A=a \mid C)\big(
  dP_B(b \mid C,D) - dP_B(b \mid C)\big) = 0$, we find that
  \begin{align*}
    \Delta(a, C, D) & = \Big|\int_{\mc{B}}
    \big(P(A=a \mid B=b,C) - P(A=a \mid C)
    \big) \big( dP_B(b \mid C,D) - dP_B(b \mid C) \big) \Big|.
  \end{align*}
  Using the fact that $|\int f(b) d\mu(b) | \le \sup_{b} \{|f(b)|\} \int
  |d\mu(b)|$ for any signed measure $\mu$ on $\mc{B}$, we conclude from
  the previous equality that for any version $P_A(\cdot \mid B,
  C)$ of the conditional probability of $A$ given $\{B, C\}$ that since
  $\int |d\mu| = \tv{\mu}$,
  \begin{align*}
    \Delta(a, C, D) & \leq 2 \sup_{b \in \mc{B}} \{|P(A=a \mid B = b, C) -
    P(A=a \mid C)|\} \tv{P_B(\cdot \mid C,D) - P_B(\cdot \mid C)}.
  \end{align*}

  Thus, to prove the lemma, it is sufficient to show\footnote{
    If $P(A = a \mid C)$ is undefined, we
    simply set it to have value $1$ and assign $P(A = a \mid B, C) = 1$ as
    well.} that for any $b \in \mc{B}$
  \begin{align}
    \label{eqn:bound-likelihood-difference}
    |P(A=a \mid B=b, C) - P(A=a \mid C)|
    \leq (e^{2 \llratio} - 1) \min\{P(A=a \mid C), P(A = a \mid C, D)\}.
  \end{align}
  To prove this upper
  bound, we consider the joint
  distribution~\eqref{eqn:joint-distribution-decomposition} and likelihood
  ratio bound~\eqref{eqn:likelihood-constraint}. The distributions
  $\{P_B(\cdot \mid A = a)\}_{a \in \mc{A}}$ are all absolutely continuous
  with respect to one another by
  assumption~\eqref{eqn:likelihood-constraint}, so it is no loss of
  generality to assume that there exists a density $p_B(\cdot \mid A = a)$ for
  which $P(B \in S \mid A = a) = \int p_B(b \mid A = a) d\mu(b)$, for some
  fixed measure $\mu$ and for which the ratio $p_B(b \mid A = a) / p_B(b \mid
  A = a') \in [e^{-\llratio}, e^\llratio]$ for all $b$. By elementary
  conditioning we have for any $S_{\mc{B}} \in \sigma(\mc{B})$ and $c \in \mc{C}$
  that
  \begin{align*}
    &P(A = a \mid B \in S_{\mc{B}}, C = c)\\
    &\qquad = \frac{P(A = a, B \in S_{\mc{B}}, C = c)}{
      P(B \in S_{\mc{B}}, C = c)} \\
    &\qquad = \frac{P(B \in S_{\mc{B}}, C = c \mid A = a) P(A = a)}{
      \sum_{a' \in \mc{A}}
      P(A = a') P(B \in S_{\mc{B}}, C = c \mid A = a')} \\
    &\qquad = \frac{P(A = a) \int_{S_{\mc{B}}} P(C = c \mid B = b, A = a)
      p_B(b \mid A = a) d\mu(b)}{
      \sum_{a' \in \mc{A}}
      P(A = a') \int_{S_{\mc{B}}} P(C = c \mid B = b, A = a')
      p_B(b \mid A = a') d\mu(b) },
  \end{align*}
  where for the last equality we used the conditional independence
  assumptions~\eqref{eqn:joint-distribution-decomposition}.
  But now we recall
  the decomposition formula~\eqref{eqn:message-distribution-factorization},
  and we can express the likelihood functions by
  \begin{align*}
    P(A=a \mid B \in S_{\mc{B}}, C = c)
    = \frac{P(A = a) \int_{S_{\mc{B}}} \fpotl_1(a, c)
      \fpotl_2(b, c) p_B(b \mid A = a) d\mu(b)}{
      \sum_{a'} P(A = a') \int_{S_{\mc{B}}}
      \fpotl_1(a', c) \fpotl_2(b, c)
      p_B(b \mid A = a') d\mu(b)}.
  \end{align*}
  As a consequence,
  there is a version of the conditional distribution of $A$ given $B$ and $C$
  such that
  \begin{equation}
    P(A=a \mid B = b, C = c)
    = \frac{P(A = a) \fpotl_1(a, c)
      p_B(b \mid A = a)}{
      \sum_{a'} P(A = a')
      \fpotl_1(a', c)
      p_B(b \mid A = a')}.
    \label{eqn:nice-density-version}
  \end{equation}
  Define the shorthand
  \begin{equation*}
    \auxratio = \frac{P(A = a)\fpotl_1(a, c)}{
      \sum_{a' \in \mc{A}} P(A=a')\fpotl_1(a', c)}.
  \end{equation*}
  We claim that
  \begin{align}
    \label{eqn:claim-auxratio-to-pratio}
    e^{-\llratio} \auxratio \le
    P(A=a \mid B = b, C = c) \leq e^\llratio \auxratio.
  \end{align}
  Assuming the correctness of bound~\eqref{eqn:claim-auxratio-to-pratio}, we
  establish inequality~\eqref{eqn:bound-likelihood-difference}. Indeed,
  $P(A=a \mid C = c)$ is a weighted average of $P(A = a \mid B = b,C = c)$, so
  we
  also have the same upper and lower bound for $P(A=a \mid C)$, that is
  \begin{equation*}
    e^{-\llratio} \auxratio \leq P(A=a \mid C) \leq e^{\llratio} \auxratio.
  \end{equation*}
  The conditional independence assumption that $A$ is independent of $D$
  given $B, C$ (recall Figure~\ref{fig:information-chaining} and the
  product~\eqref{eqn:joint-distribution-decomposition}) implies
  \begin{align*}
    P(A = a \mid C = c, D = d)
    & = \int_{\mc{B}} P(A = a \mid B = b, C = c, D = d)
    dP_B(b \mid C = c, D = d) \\
    & = \int_{\mc{B}} P(A = a \mid B = b, C = c)
    dP_B(b \mid C = c, D = d),
  \end{align*}
  and the final integrand belongs to $\auxratio [e^{-\llratio}, e^\llratio]$.
  Combining the preceding three displayed expressions, we find that
  \begin{align*}
    |P(A=a \mid B=b, C)-P(A = a \mid C)|
    & \le \big(e^{\llratio} - e^{-\llratio}\big) \auxratio \\
    & \le \big(e^{\llratio} - e^{-\llratio}\big)
    e^\llratio \min\big\{P(A = a \mid C),
    P(A = a \mid C, D)\big\}.
  \end{align*}
  This completes the proof of the upper
  bound~\eqref{eqn:bound-likelihood-difference}.

  It remains to prove inequality~\eqref{eqn:claim-auxratio-to-pratio}. We
  observe from expression~\eqref{eqn:nice-density-version} that
  \begin{equation*}
    P(A=a \mid B = b, C) =
    \frac{P(A=a) \fpotl_1(a,C)}{\sum_{a' \in \mc{A}} P(A=a')
      \fpotl_1(a',C)\frac{p_B(b \mid A=a')}{p_B(b \mid A=a)}}.
  \end{equation*}
  By the likelihood ratio bound~\eqref{eqn:likelihood-constraint}, we have
  $p_B(b \mid A = a') / p_B(b \mid A = a) \in [e^{-\llratio}, e^\llratio]$,
  and combining this with the above equation yields
  inequality~\eqref{eqn:claim-auxratio-to-pratio}.
\end{proof}


\section{Auxiliary results for
  Proposition~\ref{proposition:single-bernoulli-samples}}

In this appendix, we collect the proofs of auxiliary results involved
in the proof of
Proposition~\ref{proposition:single-bernoulli-samples}.

\subsection{Proof of Lemma~\ref{lemma:independent-message-contraction}}
\label{SecProofLemIndMesContract}
  
Let $\msg = \msg_i$; throughout the proof
we suppress the dependence on the index $i$ (and
similarly let $\sample = \sample^{(i)}$ denote a single fixed sample).  We
begin with the observation that by the chain rule for mutual
information,
\begin{equation*}
  I(\packrv; \msg) = \sum_{j=1}^d I(\packrv_j; \msg \mid \packrv_{1:j-1}).
\end{equation*}
Using the definition of mutual information and non-negativity of the
KL-divergence, we have
\begin{align*}
  I(\packrv_j; \msg \mid \packrv_{1:j-1})
  &= \E_{\packrv_{1:j-1}}
  \left[\E_\msg\left[\dkl{P_{\packrv_j}(\cdot \mid
        \msg, \packrv_{1:j-1})}{P_{\packrv_j}(\cdot
        \mid \packrv_{1:j-1})} \mid \packrv_{1:j-1}\right]\right] \\
  & \le \E_{\packrv_{1:j-1}}
  \big[\E_\msg\big[\dkl{P_{\packrv_j}(\cdot \mid
        \msg, \packrv_{1:j-1})}{P_{\packrv_j}(\cdot
        \mid \packrv_{1:j-1})}\\
      &\qquad\qquad\qquad
      ~+ \dkl{P_{\packrv_j}(\cdot \mid \packrv_{1:j-1})}{
        P_{\packrv_j}(\cdot \mid \msg, \packrv_{1:j-1})}
      \mid \packrv_{1:j-1}\big]\big].
\end{align*}

Now, we require an argument that builds off of our technical
Lemma~\ref{lemma:information-chaining}.  We claim that
Lemma~\ref{lemma:information-chaining} implies that
\begin{align}
  & \left|P(\packrv_j = \packval_j \mid \packrv_{1:j-1}, \msg)
    - P(\packrv_j = \packval_j \mid \packrv_{1:j-1})\right|
  \nonumber\\
  &\quad\quad \le 2 (e^{2 \llratio} - 1)
  \min\left\{P(\packrv_j = \packval_j \mid \packrv_{1:j-1}, \msg),
  P(\packrv_j = \packval_j \mid \packrv_{1:j-1})\right\}
  \nonumber\\
  &\quad\quad\quad ~\times \tvnorm{P_{\sample_j}(\cdot \mid \packrv_{1:j-1}, \msg)
    - P_{\sample_j}(\cdot \mid \packrv_{1:j-1})}.
  \label{eqn:tv-bound-for-independent-msgs}
\end{align}
Indeed, making the identification
\begin{equation*}
  \packrv_j \to A,
  ~~~~
  \sample_j \to B,
  ~~~~
  \packrv_{1:j-1} \to C,
  ~~~~
  \msg \to D,
\end{equation*}
the random variables satisfy the
condition~\eqref{eqn:joint-distribution-decomposition} clearly,
condition~\eqref{eqn:message-distribution-factorization} because
$\packrv_{1:j-1}$ is independent of $\packrv_j$ and $\sample_j$, and
condition~\eqref{eqn:likelihood-constraint} by construction.  This gives
inequality~\eqref{eqn:tv-bound-for-independent-msgs} by our independence
assumptions. Expanding our KL divergence bound, we have
\begin{align*}
  &\dkl{P_{\packrv_j}(\cdot \mid \msg, \packrv_{1:j-1})}{
    P_{\packrv_j}(\cdot \mid \packrv_{1:j-1})}
  + \dkl{P_{\packrv_j}(\cdot \mid \packrv_{1:j-1})}{
    P_{\packrv_j}(\cdot \mid \msg, \packrv_{1:j-1})}\\
  &\qquad\qquad = \sum_{\packval_j}
  \left(P_{\packrv_j}(\packval_j \mid \msg, \packrv_{1:j-1})
  - P_{\packrv_j}(\packval_j \mid \packrv_{1:j-1})\right)
  \log \frac{P_{\packrv_j}(\packval_j \mid \msg, \packrv_{1:j-1})}{
    P_{\packrv_j}(\packval_j \mid \packrv_{1:j-1})}.
\end{align*}
Now, using the elementary inequality for $a, b \ge 0$ that
\begin{equation*}
  \left|\log \frac{a}{b}\right|
  \le \frac{|a - b|}{\min\{a, b\}},
\end{equation*}
inequality~\eqref{eqn:tv-bound-for-independent-msgs}
implies that
\begin{align*}
  \lefteqn{\left(P_{\packrv_j}(\packval_j \mid \msg, \packrv_{1:j-1})
    - P_{\packrv_j}(\packval_j \mid \packrv_{1:j-1})\right)
    \log \frac{P_{\packrv_j}(\packval_j \mid \msg, \packrv_{1:j-1})}{
      P_{\packrv_j}(\packval_j \mid \packrv_{1:j-1})}} \\
  & \le \frac{(P_{\packrv_j}(\packval_j \mid \msg, \packrv_{1:j-1})
    - P_{\packrv_j}(\packval_j \mid \packrv_{1:j-1}))^2}{
    \min\{P_{\packrv_j}(\packval_j \mid \msg, \packrv_{1:j-1}),
    P_{\packrv_j}(\packval_j \mid \packrv_{1:j-1})\}} \\
  & \le 4 (e^{2 \llratio} - 1)^2
  \min\left\{P_{\packrv_j}(\packval_j \mid \msg, \packrv_{1:j-1}),
  P_{\packrv_j}(\packval_j \mid \packrv_{1:j-1})\right\} 
  \tvnorm{P_{\sample_j}(\cdot \mid \packrv_{1:j-1}, \msg)
    - P_{\sample_j}(\cdot \mid \packrv_{1:j-1})}^2.
\end{align*}

Substituting this into our bound on KL-divergence, we obtain
\begin{align*}
  I(\packrv_j; \msg \mid \packrv_{1:j-1})
  & = \E_{\packrv_{1:j-1}}\left[
    \E_\msg\left[\dkl{P_{\packrv_j}(\cdot \mid \msg, \packrv_{1:j-1})}{
        P_{\packrv_j}(\cdot \mid \packrv_{1:j-1})} \mid \packrv_{1:j-1}
      \right]\right] \\
  & \le 4 (e^{2\llratio} - 1)^2
  \E_{\packrv_{1:j-1}}
  \left[\E_\msg\left[
      \tvnorm{P_{\sample_j}(\cdot \mid \packrv_{1:j-1}, \msg)
        - P_{\sample_j}(\cdot \mid \packrv_{1:j-1})}^2
      \mid \packrv_{1:j-1}\right] \right].
\end{align*}
Using Pinsker's inequality, we then find that
\begin{align*}
  \lefteqn{\E_{\packrv_{1:j-1}}
    \left[\E_\msg\left[
        \tvnorm{P_{\sample_j}(\cdot \mid \packrv_{1:j-1}, \msg)
          - P_{\sample_j}(\cdot \mid \packrv_{1:j-1})}^2
        \mid \packrv_{1:j-1}\right] \right]} \\
  & \le \half \E_{\packrv_{1:j-1}}
  \left[\E_\msg\left[\dkl{P_{\sample_j}(\cdot \mid
        \msg, \packrv_{1:j-1})}{
        P_{\sample_j}(\cdot \mid \packrv_{1:j-1})}
      \mid \packrv_{1:j-1}\right]\right]
  = \half I(\sample_j; \msg \mid \packrv_{1:j-1}).
\end{align*}
In particular, we have
\begin{equation}
  \label{eqn:index-contraction}
  I(\packrv_j; \msg \mid \packrv_{1:j-1})
  \le 2\left(e^{2 \llratio} - 1\right)^2 I(\sample_j; \msg \mid
  \packrv_{1:j-1}).
\end{equation}

Lastly, we argue that $I(\sample_j; \msg \mid \packrv_{1:j-1}) \le
I(\sample_j; \msg \mid \sample_{1:j-1})$. Indeed, we have by
definition that
\begin{align*}
  I(\sample_j; \msg \mid \packrv_{1:j-1})
  & \stackrel{(i)}{=} H(\sample_j)
  - H(\sample_j \mid \msg, \packrv_{1:j-1}) \\
  & \stackrel{(ii)}{\le}
  H(\sample_j) - H(\sample_j \mid \msg, \packrv_{1:j-1}, \sample_{1:j-1}) \\
  & \stackrel{(iii)}{=}
  H(\sample_j \mid \sample_{1:j-1})
  - H(\sample_j \mid \msg, \sample_{1:j-1})
  = I(\sample_j ; \msg \mid \sample_{1:j-1}).
\end{align*}
Here, equality $(i)$ follows since $\sample_j$ is independent of
$\packrv_{1:j-1}$, inequality $(ii)$ because conditioning reduces entropy,
and equality $(iii)$ because $\sample_j$ is independent of $\sample_{1:j-1}$.
Thus
\begin{equation*}
  I(\packrv; \msg)
  = \sum_{j=1}^d I(\packrv_j; \msg \mid \packrv_{1:j-1})
  \le 2 (e^{2 \llratio} - 1)^2
  \sum_{j=1}^d I(\sample_j; \msg \mid \sample_{1:j-1})
  = 2 (e^{2 \llratio} - 1)^2
  I(\sample; \msg),
\end{equation*}
which completes the proof.

\subsection{Proof of Lemma~\ref{LemTensorize}}
\label{AppTensorize}
By assumption, the message $\msg_i$ is constructed based only on
$\sample^{(i)}$.  Therefore, we have
\begin{align*}
I(\packrv; \msg_{1:\nummac}) = \sum_{i=1}^\nummac I(\packrv; \msg_i
\mid \msg_{1:i-1}) & = \sum_{i=1}^\nummac H(\msg_i \mid \msg_{1:i-1})
- H(\msg_i \mid \packrv, \msg_{1:i-1}) \nonumber \\
& \le \sum_{i=1}^\nummac H(\msg_i) - H(\msg_i \mid \packrv,
\msg_{1:i-1}) \\
& = \sum_{i=1}^\nummac H(\msg_i) - H(\msg_i \mid \packrv) =
  \sum_{i=1}^\nummac I(\packrv; \msg_i)
\end{align*}
where we have used that conditioning reduces entropy and $\msg_i$ is
conditionally independent of $\msg_{1:i-1}$ given $\packrv$.

\subsection{Proof of inequality~\eqref{eqn:tv-norm-to-info}}
\label{AppSmall}

Let $\statprob_\packval$ be shorthand for $\statprob_\msg(\cdot \mid
\packrv = \packval)$. The triangle inequality implies that
\begin{equation*}
  \tvnorm{\statprob_\packval - \statprob_\altpackval} \le
  \tvnorm{\statprob_\packval - (1/2)(\statprob_\packval +
    \statprob_\altpackval)} + \half \tvnorm{\statprob_\packval -
    \statprob_\altpackval},
\end{equation*}
and similarly swapping the roles of $\altpackval$ and $\packval$, whence
\begin{equation*}
  \tvnorm{\statprob_\packval - \statprob_\altpackval} \le 2
  \min\{\tvnorm{\statprob_\packval - (1/2) (\statprob_\altpackval +
    \statprob_\packval)}, \tvnorm{\statprob_\altpackval - (1/2)
    (\statprob_\altpackval + \statprob_\packval)}\}.
\end{equation*}
By Pinsker's inequality, we thus have the upper bound
\begin{align*}
  \tvnorm{\statprob_\packval - \statprob_\altpackval}^2 & \le 2
  \min\{\dkl{\statprob_\packval}{(1/2)(\statprob_\packval +
    \statprob_\altpackval)},
  \dkl{\statprob_\altpackval}{(1/2)(\statprob_\packval +
    \statprob_\altpackval)}\} \\ & \le
  \dkl{\statprob_\packval}{(1/2)(\statprob_\packval +
    \statprob_\altpackval)} +
  \dkl{\statprob_\altpackval}{(1/2)(\statprob_\packval +
    \statprob_\altpackval)} = 2 I(\msg; \packrv)
\end{align*}
by the definition of mutual information.


\section{Auxiliary results for Theorem~\ref{theorem:gaussian-mean-communication}}

In this appendix, we collect the proofs of auxiliary results involved
in the proof of Theorem~\ref{theorem:gaussian-mean-communication}.

\subsection{Proof of Lemma~\ref{lemma:independent-message-unbounded}}
\label{SecProofLemIndMesUnb}

This proof is similar to that
Lemma~\ref{lemma:independent-message-contraction}, but we must be
careful when conditioning on events of the form $\sample_j^{(i)} \in
\llset_j$. For notational simplicity, we again suppress all dependence
of $\sample$ and $\msg$ on the machine index $i$.
Our goal is to prove that
\begin{equation}
  \label{eqn:index-contraction-gaussian}
  I(\packrv_j; \msg \mid \packrv_{1:j-1})
  \le H(E_j) + P(E_j = 0)
  + 2\big(e^{4 \llratio} - 1\big)^2
  I(\sample_j; \msg \mid \packrv_{1:j-1}).
\end{equation}
Up to the additive terms, this is equivalent to the earlier
bound~\eqref{eqn:index-contraction} in the proof of
Lemma~\ref{lemma:independent-message-contraction}, so that proceeding
\emph{mutatis mutandis} completes the proof. We now turn to proving
inequality~\eqref{eqn:index-contraction-gaussian}.

We begin by noting that $I(X; Y \mid Z) \le I(X, W; Y \mid Z)$ for any random
variables $W, X, Y, Z$, because conditioning reduces entropy:
\begin{equation}
  I(X; Y \mid Z) = H(Y \mid Z) - H(Y \mid X, Z)
  \le H(Y \mid Z) - H(Y \mid W, X, Z)
  = I(X, W; Y \mid Z).
  \label{eqn:increase-information}
\end{equation}
As a consequence, recalling the random variable $E_j$ (the indicator of
$\sample_j \in \llset_j$), we have
\begin{align}
  I(\packrv_j; \msg \mid \packrv_{1:j-1})
  \le I(\packrv_j; \msg, E_j \mid \packrv_{1:j-1})
  & = I(\packrv_j; \msg \mid E_j, \packrv_{1:j-1})
  + I(\packrv_j; E_j \mid \packrv_{1:j-1}) \nonumber \\
  & \le I(\packrv_j; \msg \mid E_j, \packrv_{1:j-1})
  + H(E_j \mid \packrv_{1:j-1}) \nonumber \\
  & = I(\packrv_j; \msg \mid E_j, \packrv_{1:j-1}) + H(E_j),
  \label{eqn:calling-me-home}
\end{align}
where the final equality follows because $E_j$ is independent of
$\packrv_{1:j-1}$. 
Comparing to inequality~\eqref{eqn:index-contraction-gaussian},
we need only control the first term in the bound~\eqref{eqn:calling-me-home}.

To that end,
note that given $E_j$, the variable $\packrv_j$ is
independent of $\packrv_{1:j-1}$, $\sample_{1:j-1}$, $\packrv_{j+1:d}$,
and $\sample_{j+1:d}$. Moreover,
by the assumption in the lemma we have for any $S \in \sigma(\llset_j)$
that
\begin{equation*}
  \frac{P_{\sample_j}(S \mid \packrv = \packval, E_j =
    1)}{P_{\sample_j}(S \mid \packrv = \packval', E_j = 1)} =
  \frac{P_{\sample_j}(S \mid \packrv = \packval)}{
    P_{\sample_j}(\sample_j \in \llset_j \mid \packrv = \packval)}
  \frac{P_{\sample_j}(\sample_j \in \llset_j \mid \packrv = \packval')}{
    P_{\sample_j}(\sample_j \in S \mid \packrv = \packval')} \le \exp(2
  \llratio).
\end{equation*}
Applying Lemma~\ref{lemma:information-chaining} yields that the
difference
\begin{align*}
  \Delta_j \defeq P(\packrv_j = \packval_j \mid \packrv_{1:j-1}, \msg, E_j =
  1) - P(\packrv_j = \packval_j \mid \packrv_{1:j-1}, E_j = 1)
\end{align*}
is bounded as
\begin{align*}
  |\Delta_j| & \le 2\big(e^{4 \llratio} - 1\big) \tvnorm{P_{\sample_j}(\cdot
  \mid \packrv_{1:j-1}, \msg, E_j = 1) - P_{\sample_j}(\cdot \mid
  \packrv_{1:j-1}, E_j = 1)} \\
  & \qquad \quad ~
  \times \min\big\{P(\packrv_j = \packval_j \mid \packrv_{1:j-1}, \msg, E_j = 1),
  P(\packrv_j = \packval_j \mid \packrv_{1:j-1}, E_j = 1)\big\}
  \nonumber
\end{align*}
(cf.\ the inequality~\eqref{eqn:tv-bound-for-independent-msgs} in
the proof of Lemma~\ref{lemma:independent-message-contraction}).
Proceeding as in the proof of
Lemma~\ref{lemma:independent-message-contraction}, this
expression leads to
the bound
\begin{align}
  \label{eqn:index-contraction-v2}
  I(\packrv_j; \msg \mid \packrv_{1:j-1}, E_j = 1) & \leq 2 \big( e^{4
    \llratio} - 1\big)^2 I(\sample_j; \msg \mid \packrv_{1:j-1}, E_j =
  1).
\end{align}
By the definition of conditional mutual information,
\begin{align*}
  I(\packrv_j; \msg \mid E_j, \packrv_{1:j-1})
  & = P(E_j = 1) I(\packrv_j; \msg \mid \packrv_{1:j-1}, E_j = 1)
  + P(E_j = 0) I(\packrv_j; \msg \mid \packrv_{1:j-1}, E_j = 0) \nonumber \\
  & \le I(\packrv_j; \msg \mid \packrv_{1:j-1}, E_j = 1)
  + P(E_j = 0) \log 2,
\end{align*}
where the inequality follows because $\packrv_j \in \{-1, 1\}$.  But combining
this inequality with the bounds~\eqref{eqn:index-contraction-v2}
and~\eqref{eqn:calling-me-home} gives the desired
result~\eqref{eqn:index-contraction-gaussian}.


\subsection{Proof of Lemma~\ref{lemma:gaussian-information-bounds}}

In order to prove inequality~\eqref{eqn:markov-bound-application}, we
note that $\packrv\to \sample^{(i)}\to \msg_i$ forms a Markov
chain. Thus, the classical data-processing inequality~\cite{Cover2006} implies
that
\begin{align*}
  I(\packrv; \msg_i) \leq I(\packrv; \sample^{(i)})
  \le \sum_{k=1}^{\numobs_i} I(\packrv; \sample^{(i,k)}).
\end{align*}
Let $\statprob_\packval$ denote the conditional distribution of
$\sample^{(i,k)}$ given $\packrv = \packval$.  Then the convexity of the
KL-divergence establishes inequality~\eqref{eqn:markov-bound-application}
via
\begin{equation*}
  I(\packrv; \sample^{(i,k)})
  \le \frac{1}{|\packset|^2}
  \sum_{\packval, \altpackval \in \packset}
  \dkl{\statprob_\packval}{\statprob_\altpackval}
  = \frac{\delta^2}{2 \sigma^2} \frac{1}{|\packset|^2}
  \sum_{\packval, \altpackval \in \packset} \ltwo{\packval - \altpackval}^2
  = \frac{d \delta^2}{\sigma^2}.
\end{equation*}

To prove inequality~\eqref{eqn:transition-bound-application}, we apply
Lemma~\ref{lemma:independent-message-unbounded}.  First,
consider two one-dimensional normal distributions, each with
$\numobs_i$ independent observations and variance $\sigma^2$, but where
one has mean $\delta$ and the other mean $-\delta$. For
fixed $a \ge 0$, the ratio of their
densities is
\begin{equation*}
  \frac{\exp(-\frac{1}{2\sigma^2} \sum_{l=1}^{\numobs_i} (x_l - \delta)^2)}{
    \exp(-\frac{1}{2\sigma^2} \sum_{l=1}^{\numobs_i} (x_l + \delta)^2)}
  = \exp\bigg(\frac{\delta}{\sigma^2} \sum_{l=1}^{\numobs_i} x_l\bigg)
  \le \exp\bigg(\frac{\sqrt{\numobs_i}\delta a}{\sigma^2}\bigg)
\end{equation*}
whenever $|\sum_l x_l| \le \sqrt{\numobs_i}a$. As a consequence, we see that
by taking the sets
\begin{equation*}
  \llset_j = \bigg\{x \in \R^{\numobs_i}
  : \bigg|\sum_{l=1}^{\numobs_i} x_l\bigg|
  \le \sqrt{\numobs_i}a\bigg\},
\end{equation*}
we satisfy the conditions of
Lemma~\ref{lemma:independent-message-unbounded} with the
quantity $\llratio$ defined as $\llratio =
\sqrt{\numobs_i}\delta a / \sigma^2$. In addition, when $\llratio \le 1.2564$,
we have $\exp(\llratio) - 1 \le 2\llratio$, so under the conditions of the
lemma, $\exp(4 \llratio) - 1 = \exp(4\sqrt{\numobs_i} \delta a / \sigma^2) - 1
\le 8\sqrt{\numobs_i} \delta a / \sigma^2$. Recalling the definition of the
indicator random
variable $E_j = \sindic{\sample^{(i)}_j \in \llset_j}$ from
Lemma~\ref{lemma:independent-message-unbounded}, we obtain
\begin{equation}
  I(\packrv; \msg_i)
  \le 128 \frac{\delta^2 a^2}{\sigma^4}
  \numobs_i I(\sample^{(i)}; \msg_i)
  + \sum_{j=1}^d H(E_j)
  + \sum_{j=1}^d P(E_j = 0).
  \label{eqn:almost-done-independents}
\end{equation}
Comparing this inequality with
inequality~\eqref{eqn:transition-bound-application}, we see that we must
bound the probability of the event $E_j = 0$.

Bounding $P(E_j = 0)$ is not challenging, however. From standard Gaussian tail
bounds, we have for $Z_l$ distributed i.i.d.\ according to
$\normal(\delta, \sigma^2)$ that
\begin{align}
  P(E_j = 0)
  & = P\bigg(\bigg|\sum_{l=1}^{\numobs_i} Z_l\bigg| \ge \sqrt{\numobs_i}a
  \bigg) \nonumber \\
  & = P\bigg(\sum_{l=1}^{\numobs_i} (Z_l - \delta)
  \ge \sqrt{\numobs_i}a - \numobs \delta\bigg)
  + P\bigg(\sum_{l=1}^{\numobs_i} (Z_l - \delta)
  \le -\sqrt{\numobs_i}a - \numobs \delta\bigg) \nonumber \\
  & \le 2 \exp\bigg(-\frac{(a - \sqrt{\numobs_i} \delta)^2}{2 \sigma^2}
  \bigg).
  \label{eqn:gaussian-tail-bounds}
\end{align}
Since $h_2(p) \le h_2(\half)$ and $I(\packrv; \msg_i)
\le d \log 2$ regardless, this
provides the bounds on the entropy and probability terms in
inequality~\eqref{eqn:almost-done-independents}
to yield the result~\eqref{eqn:transition-bound-application}.


\subsection{Proof of Lemma~\ref{LemIceAxe}}
\label{AppLemIceAxe}

Combining inequalities~\eqref{eqn:markov-bound-application}
and~\eqref{eqn:transition-bound-application} yields
\begin{equation}
  \label{eqn:gaussian-information-contraction}
  \begin{split}
    I(\packrv; \msg_i) & \le \frac{\numobs_i \delta^2}{\sigma^2}
    \min\bigg\{ 128 \frac{a^2}{\sigma^2} H(\msg_i), d\bigg\} + d \,
    h_2\bigg( \!\min\Big\{ 2 \exp\bigg(-\frac{(a - \sqrt{\numobs_i}
      \delta)^2}{ 2 \sigma^2}\bigg), \half \Big\}\bigg)
    \\ & \qquad ~ +
    2d \exp\bigg(-\frac{(a - \sqrt{\numobs_i} \delta)^2}{2
      \sigma^2}\bigg),
  \end{split}
\end{equation}
true for all $a, \delta \ge 0$ and $\numobs_i, \sigma^2$ such that
$\sqrt{\numobs_i} a \delta \le 1.2564 \sigma^2 / 4$
and $a \ge \delta \sqrt{\numobs_i}$.

Now, we consider each of the terms in the bound in
inequality~\eqref{eqn:gaussian-information-contraction} in turn, finding
settings of $\delta$ and $a$ so that each term is small. Let us set $a = 4
\sigma \sqrt{\log \nummac}$.  We begin with the third term in the
bound~\eqref{eqn:gaussian-information-contraction}, where we note that by
definining $\delta_3$ as the positive root of
\begin{equation}
  \label{eqn:delta-3}
  \delta_3^2 \defeq \frac{\sigma^2 }{16 \cdot 16\log(\nummac)\max_i \numobs_i},
\end{equation}
then for $0 \le \delta \le \delta_3$ the conditions $\frac{\sqrt{\numobs_i} a
  \delta}{\sigma^2} \le \frac{1.2564}{4}$ and $\sqrt{\numobs_i} \delta \le a$
in Lemma~\ref{lemma:gaussian-information-bounds} are satisfied. In addition,
we have $(a - \sqrt{\numobs_i} \delta)^2 \ge (4 - 1/256)^2\sigma^2 \log
\nummac \ge 15 \sigma^2 \log \nummac$ for $0 \le \delta \le \delta_3$, so for
such $\delta$
\begin{equation*}
  \sum_{i=1}^\nummac 2 \exp\bigg(-\frac{(a - \sqrt{\numobs_i} \delta)^2}{
    2 \sigma^2} \bigg)
  \le 2 \nummac \exp(-(15/2) \log \nummac)
  = \frac{2}{\nummac^{15/2}}
  < 2 \cdot 10^{-5}.
\end{equation*}
Secondly, we have $h_2(q) \le (6/5) \sqrt{q}$ for $q \ge 0$. As a
consequence, we see that for $\delta_2$ chosen identically to the
choice~\eqref{eqn:delta-3} for $\delta_3$, we have
\begin{equation*}
  \sum_{i=1}^\nummac 2 h_2
  \bigg(2 \exp\bigg(-\frac{(a - \sqrt{\numobs_i} \delta_2)^2}{2
    \sigma^2}\bigg)\bigg)
  \le \frac{12 \nummac}{5}\sqrt{2}
  \exp(-(15/4) \log \nummac)
  < \frac{2}{49}.
\end{equation*}
In particular, with the choice $a = 4 \sigma \sqrt{\log \nummac}$ and for all
$0 \le \delta \le \delta_3$,
inequality~\eqref{eqn:gaussian-information-contraction} implies 
the desired bound~\eqref{eqn:theorem-1-mutual-information-final-bound}.


\section{Auxiliary results for Theorem~\ref{theorem:interactive-gaussian}}

In this appendix, we collect the proofs of auxiliary results for
Theorem~\ref{theorem:interactive-gaussian}.

\subsection{Proof of Lemma~\ref{lemma:information-contraction-in-subset}}
\label{sec:proof-contraction-subset}

We state an intermediate claim from which
Lemma~\ref{lemma:information-contraction-in-subset} follows
quickly. Let us temporarily assume that the set $\llset$ in the
statement of the lemma is $\llset = \range(\sample^{(i)})$, so that
there is no restriction on the distributions
$\statprob_{\sample^{(i)}}$, that is, the likelihood ratio
bound~\eqref{eqn:likelihood-ratio-subset-bound} holds for all
measurable sets $S$. We claim that in this case,
\begin{equation}
  I(\packrv; \msg) \le 2\big(e^{2 \llratio} - 1\big)^2 I(\sample;
  \msg).
  \label{eqn:interactive-pack-to-message}
\end{equation}
Assuming that we have established
inequality~\eqref{eqn:interactive-pack-to-message}, the proof of
Lemma~\ref{lemma:information-contraction-in-subset} follows,
\emph{mutatis mutandis}, as in the proof of
Lemma~\ref{lemma:independent-message-unbounded} from
Lemma~\ref{lemma:independent-message-contraction}.

Let us now prove the claim~\eqref{eqn:interactive-pack-to-message}.
By the chain-rule for mutual information, we have
\begin{align*}
  I(\packrv; \msg) = \sum_{t = 1}^\nummsg I(\packrv; \msg_t \mid
  \msg_{1:t-1}).
\end{align*}
Let $P_{\msg_t}(\cdot \mid \msg_{1:t-1})$ denote the (marginal)
distribution of $\msg_t$ given $\msg_{1:t-1}$ and define
$P_\packrv(\cdot \mid \msg_{1:t})$ to be the distribution of $\packrv$
conditional on $\msg_{1:t}$. Then we have by marginalization that
\begin{equation*}
  P_\packrv(\cdot \mid \msg_{1:t-1})
  = \int P_\packrv(\cdot \mid \msg_{1:t-1}, \lcmsg_t)
  dP_{\msg_t}(\lcmsg_t \mid \msg_{1:t-1})
\end{equation*}
and thus
\begin{equation*}
  I(\packrv; \msg_t \mid \msg_{1:t-1})
  = \E_{\msg_{1:t-1}} \big[\E_{\msg_t}\big[
      \dkl{P_\packrv(\cdot \mid \msg_{1:t})}{P_\packrv(\cdot \mid \msg_{1:t-1})}
      \mid \msg_{1:t-1}\big]\big].
\end{equation*}
We now bound the above KL divergence using the assumed likelihood ratio
bound on $\statprob_\sample$ in the
lemma (when $\llset = \mc{\sample}$, the entire sample space).

By the nonnegativity of the KL divergence, we have
\begin{align*}
\dkl{P_\packrv(\cdot \mid \msg_{1:t})}{P_\packrv(\cdot \mid
  \msg_{1:t-1})}& \le \dkl{P_\packrv(\cdot \mid
  \msg_{1:t})}{P_\packrv(\cdot \mid \msg_{1:t-1})} +
\dkl{P_\packrv(\cdot \mid \msg_{1:t-1})}{ P_\packrv(\cdot \mid
  \msg_{1:t})} \\
& = \sum_{\packval \in \packset} \big(p_\packrv(\packval \mid
\msg_{1:t-1}) - p_\packrv(\packval \mid \msg_{1:t}) \big) \log
\frac{p_\packrv(\packval \mid \msg_{1:t-1})}{ p_\packrv(\packval \mid
  \msg_{1:t})}
\end{align*}
where $p_\packrv$ denotes the p.m.f.\ of $\packrv$.

Next we claim that the difference $\Delta_t \defeq \big|p_\packrv(\packval
\mid \msg_{1:t-1}) - p_\packrv(\packval \mid \msg_{1:t})\big|$ is
upper bounded as
\begin{align}
  \label{eqn:tv-bound-claim}
  |\Delta_t| & \leq 2 \big( e^{2\numobs\llratio} - 1\big)
  \min\big\{p_\packrv(\packval \mid \msg_{1:t-1}), p_\packrv(\packval
  \mid \msg_{1:t})\big\} \tvnorm{P_{\sample^{(i_t)}}(\cdot \mid
    \msg_{1:t}) - P_{\sample^{(i_t)}}(\cdot \mid \msg_{1:t-1})}.
\end{align}
Deferring the proof of this claim to the end of this section, we give
the remainder of the proof.  First, by a first-order convexity
argument, we have
\begin{align*}
  |\log a - \log b| & \le \frac{|a - b|}{\min\{a, b\}} \quad \mbox{for
    any $a, b > 0$.}
\end{align*}
Combining this bound with inequality~\eqref{eqn:tv-bound-claim} yields
\begin{align*}
  \lefteqn{\Delta_t \log \frac{p_\packrv(\packval \mid \msg_{1:t-1})}{
      p_\packrv(\packval \mid \msg_{1:t})}
    \leq \frac{\Delta_t^2}{
      \min\{p_\packrv(\packval \mid \msg_{1:t-1}), p_{\packrv}(\packval
      \mid \msg_{1:t})\}}} \\
  & \qquad ~ \le 4\big(e^{2 \numobs\llratio} - 1\big)^2
  \min\big\{p_\packrv(\packval \mid \msg_{1:t-1}), p_\packrv(\packval
  \mid \msg_{1:t})\big\} \tvnorm{P_{\sample^{(i_t)}}(\cdot \mid
    \msg_{1:t}) - P_{\sample^{(i_t)}}(\cdot \mid \msg_{1:t-1})}^2.
\end{align*}
Since $p_\packrv$ is a p.m.f.,
we have the following upper bound on the symmetrized KL divergence
between $P_\packrv(\cdot \mid \msg_{1:t})$ and $P_\packrv(\cdot \mid \msg_{1:t-1})$:
\begin{align*}
  \lefteqn{\dkl{P_\packrv(\cdot \mid \msg_{1:t})}{ P_\packrv(\cdot \mid
      \msg_{1:t-1})} + \dkl{P_\packrv(\cdot \mid \msg_{1:t-1})}{
      P_\packrv(\cdot \mid \msg_{1:t})}} \\
  & \le 4\big(e^{2 \numobs\llratio} - 1\big)^2
  \tvnorm{P_{\sample^{(i_t)}}(\cdot \mid \msg_{1:t}) -
    P_{\sample^{(i_t)}}(\cdot \mid \msg_{1:t-1})}^2 \sum_{\packval \in
    \packset} \min\big\{p_\packrv(\packval \mid \msg_{1:t-1}),
  p_\packrv(\packval \mid \msg_{1:t})\big\} \\ & \le 4\big(e^{2
    \numobs\llratio} - 1\big)^2 \tvnorm{P_{\sample^{(i_t)}}(\cdot \mid
    \msg_{1:t}) - P_{\sample^{(i_t)}}(\cdot \mid \msg_{1:t-1})}^2 \\
  & \leq \half \: \dkl{P_{\sample^{(i_t)}}(\cdot \mid \msg_{1:t})}{
    P_{\sample^{(i_t)}}(\cdot \mid \msg_{1:t-1})},
\end{align*}
where the final step follows from Pinsker's inequality.  Taking
expectations, we have
\begin{align*}
\half \E_{\msg_{1:t-1}}\big[\E_{\msg_t}
  \big[\dkl{P_{\sample^{(i_t)}}(\cdot \mid \msg_{1:t})}{
      P_{\sample^{(i_t)}}(\cdot \mid \msg_{1:t-1})} \mid \msg_{1:t-1}
    \big]\big] = \half I(\sample^{(i_t)}; \msg_t \mid \msg_{1:t-1}).
\end{align*}
Finally, because conditioning reduces
entropy (recall inequality~\eqref{eqn:increase-information}), we have
\begin{align*}
  I(\sample^{(i_t)}; \msg_t \mid \msg_{1:t-1})
  \le I(\sample; \msg_t \mid \msg_{1:t-1}).
\end{align*}
By
the chain rule for mutual information, we have $\sum_{t = 1}^T
I(\sample; \msg_t \mid \msg_{1:t-1}) = I(\sample; \msg)$, so the
proof is complete.


\paragraph{Proof of inequality~\eqref{eqn:tv-bound-claim}}

It remains to prove inequality~\eqref{eqn:tv-bound-claim}: in order to
do so, we establish a one-to-one correspondence between the variables
in Lemma~\ref{lemma:information-chaining} and the variables in
inequality~\eqref{eqn:tv-bound-claim}.  Let us begin by making the
identifications
\begin{equation*}
  \packrv \rightarrow A ~~~~~~~ \sample^{(i_t)} \rightarrow B ~~~~~~~
  \msg_{1:t-1} \rightarrow C ~~~~~~~ \msg_t \rightarrow D.
\end{equation*}
For Lemma~\ref{lemma:information-chaining} to hold, we must verify
conditions~\eqref{eqn:joint-distribution-decomposition},
\eqref{eqn:message-distribution-factorization},
and~\eqref{eqn:likelihood-constraint}. For
condition~\eqref{eqn:joint-distribution-decomposition} to hold, $\msg_{t}$
must be independent of $\packrv$ given
$\{\msg_{1:t-1},\sample^{(i_t)}\}$. Since the distribution of
$P_{\msg_{t}}(\cdot \mid \msg_{1:t-1}, \sample^{(i_t)})$ is
measurable-$\{\msg_{1:t-1},\sample^{(i_t)}\}$,
condition~\eqref{eqn:likelihood-constraint} is satisfied by the assumption in
Lemma~\ref{lemma:information-contraction-in-subset}.

Finally, for condition~\eqref{eqn:message-distribution-factorization} to hold,
we must be able to factor the conditional probability
of $\msg_{1:t-1}$ given $\{\packrv, \sample^{(i_t)}\}$ as
\begin{align}
  \label{eqn:examine-decomposition}
  P(\msg_{1:t-1} = \lcmsg_{1:t-1} \mid \packrv, \sample^{(i_t)}) =
  \fpotl_1(\packrv, \lcmsg_{1:t-1}) \fpotl_2(\sample^{(i_t)},
  \lcmsg_{1:t-1}).
\end{align}
To prove this decomposition, notice that
\begin{equation*}
  P(\msg_{1:t-1} = \lcmsg_{1:t-1} \mid \packrv, \sample^{(i_t)}) =
  \prod_{k=1}^{t-1} P(\msg_{k} = \lcmsg_k \mid \msg_{1:k-1}, \packrv,
  \sample^{(i_t)}).
\end{equation*}
For any $k\in\{1,\dots,t-1\}$, if $i_k = i_t$---that is, the message
$\msg_k$ is generated based on sample $\sample^{(i_t)} =
\sample^{(i_k)}$---then $\msg_k$ is independent of $\packrv$ given
$\{\sample^{(i_t)}, \msg_{1:k-1}\}$. Thus, $P_{\msg_{k}}(\cdot \mid
\msg_{1:k-1},\packrv,\sample^{(i_t)})$ is
measurable-$\{\sample^{(i_t)},\msg_{1:k-1}\}$.  If the $k$th index
$i_k \neq i_t$, then $\msg_k$ is independent of $\sample^{(i_t)}$
given $\{\msg_{1:k-1},\packrv\}$ by construction, which means
$P_{\msg_{k}}(\cdot \mid \msg_{1:k-1},\packrv,\sample^{(i_t)}) =
P_{\msg_k}(\cdot \mid \msg_{1:k-1}, \packrv)$, thereby verifying the
decomposition~\eqref{eqn:examine-decomposition}.  Thus, we have
verified that each of the conditions of
Lemma~\ref{lemma:information-chaining} holds, so that
inequality~\eqref{eqn:tv-bound-claim} follows.

\subsection{Proof of Lemma~\ref{lemma:interactive-gaussian-bounds}}
\label{sec:proof-interactive-gaussian-information}

To prove Lemma~\ref{lemma:interactive-gaussian-bounds},
fix an arbitrary realization $\packval_{\setminus j} \in \{-1, 1\}^{d-1}$ of
$\packrv_{\setminus j}$.  Conditioning on $\packrv_{\setminus j} =
\packval_{\setminus j}$, note that $\packval_j \in \{-1, 1\}$, and consider
the distributions of the $j$th coordinate of each (local) sample
$\sample^{(i)}_j \in \R^\numobs$,
\begin{equation*}
  P_{\sample^{(i)}_j}(\cdot \mid \packrv_j = \packval_j,
  \packrv_{\setminus j} = \packval_{\setminus j})
  ~~~ \mbox{and} ~~~
  P_{\sample^{(i)}_j}(\cdot \mid \packrv_j = -\packval_j,
  \packrv_{\setminus j} = \packval_{\setminus j}).
\end{equation*}
We claim that these distributions---with appropriate
constants---satisfy the conditions of
Lemma~\ref{lemma:information-contraction-in-subset}.  Indeed, fix $a
\ge 0$, take the set $\llset = \{x \in \R^\numobs \mid \lone{x} \le
\sqrt{\numobs} a\}$, and set the log-likelihood ratio parameter
$\llratio = \sqrt{\numobs} \delta a / \sigma^2$. Then the random
variable $E_j = 1$ if $\sample_j^{(i)} \in \llset$ for all $i = 1,
\ldots, \nummac$, and we may apply precisely proof
of
Lemma~\ref{lemma:information-contraction-in-subset} (we still obtain the
factorization~\eqref{eqn:examine-decomposition} by conditioning
everything on $\packrv_{\setminus j} = \packval_{\setminus j}$). Thus
we obtain
\begin{equation}
  \begin{split}
    I(\packrv_j; \msg \mid \packrv_{\setminus j} = \packval_{\setminus j})
    & \le 2 \big(e^{4 \llratio} - 1\big)^2 I(\sample_j; \msg \mid
    \packrv_{\setminus j} = \packval_{\setminus j}) \\ & \qquad ~ + H(E_j
    \mid \packrv_{\setminus j} = \packval_{\setminus j}) + P(E_j = 0 \mid
    \packrv_{\setminus j} = \packval_{\setminus j}).
  \end{split}
  \label{eqn:conditional-contraction-subset}
\end{equation}
Of course, the event $E_j$ is independent of $\packrv_{\setminus j}$
by construction, so that $P(E_j = 0 \mid \packrv_{\setminus j}) =
P(E_j = 0)$ and $H(E_j \mid \packrv_{\setminus j} =
\packval_{\setminus j}) = H(E_j)$, and standard Gaussian tail bounds
(cf.\ the proof of Lemma~\ref{lemma:gaussian-information-bounds} and
inequality~\eqref{eqn:gaussian-tail-bounds}) imply that
\begin{equation*}
  H(E_j) \le \nummac h_2\bigg(2 \exp\bigg(-\frac{(a - \sqrt{n} \delta)^2}{
    2 \sigma^2}\bigg)\bigg)
  ~~~ \mbox{and} ~~~
  P(E_j = 0) \le 2 \nummac \exp\bigg(-\frac{(a - \sqrt{n} \delta)^2}{
    2 \sigma^2}\bigg).
\end{equation*}
Thus by integrating over $\packrv_{\setminus j} = \packval_{\setminus j}$,
inequality~\eqref{eqn:conditional-contraction-subset}
implies the lemma.


\bibliographystyle{abbrvnat}
\bibliography{bib}

\begin{thebibliography}{35}
\providecommand{\natexlab}[1]{#1}
\providecommand{\url}[1]{\texttt{#1}}
\expandafter\ifx\csname urlstyle\endcsname\relax
  \providecommand{\doi}[1]{doi: #1}\else
  \providecommand{\doi}{doi: \begingroup \urlstyle{rm}\Url}\fi

\bibitem[Abelson(1980)]{Abelson1980}
H.~Abelson.
\newblock Lower bounds on information transfer in distributed computations.
\newblock \emph{Journal of the ACM}, 27\penalty0 (2):\penalty0 384--392, 1980.

\bibitem[Balcan et~al.(2012)Balcan, Blum, Fine, and Mansour]{Balcan2012}
M.-F. Balcan, A.~Blum, S.~Fine, and Y.~Mansour.
\newblock Distributed learning, communication complexity and privacy.
\newblock In \emph{Proceedings of the Twenty Fifth Annual Conference on
  Computational Learning Theory}, 2012.
\newblock URL \url{http://arxiv.org/abs/1204.3514}.

\bibitem[Ball(1997)]{Ball1997}
K.~Ball.
\newblock An elementary introduction to modern convex geometry.
\newblock In S.~Levy, editor, \emph{Flavors of Geometry}, pages 1--58. MSRI
  Publications, 1997.

\bibitem[Bar-Yossef et~al.(2004)Bar-Yossef, Jayram, Kumar, and
  Sivakumar]{Bar-Yossef2004}
Z.~Bar-Yossef, T.~S. Jayram, R.~Kumar, and D.~Sivakumar.
\newblock An information statistics approach to data stream and communication
  complexity.
\newblock \emph{Journal of Computer and System Sciences}, 68\penalty0
  (4):\penalty0 702--732, 2004.

\bibitem[Barak et~al.(2010)Barak, Braverman, Chen, and Rao]{Barak2010}
B.~Barak, M.~Braverman, X.~Chen, and A.~Rao.
\newblock How to compress interactive communication.
\newblock In \emph{Proceedings of the Fourty-Second Annual ACM Symposium on the
  Theory of Computing}, 2010.

\bibitem[Birg\'e(1983)]{Birge1983}
L.~Birg\'e.
\newblock Approximation dans les espaces m\'etriques et th\'eorie de
  l'estimation.
\newblock \emph{Zeitschrift f\"ur Wahrscheinlichkeitstheorie und verwebte
  Gebiet}, 65:\penalty0 181--238, 1983.

\bibitem[Boyd and Vandenberghe(2004)]{Boyd2004}
S.~Boyd and L.~Vandenberghe.
\newblock \emph{Convex Optimization}.
\newblock Cambridge University Press, 2004.

\bibitem[Boyd et~al.(2011)Boyd, Parikh, Chu, Peleato, and Eckstein]{Boyd2011}
S.~Boyd, N.~Parikh, E.~Chu, B.~Peleato, and J.~Eckstein.
\newblock Distributed optimization and statistical learning via the alternating
  direction method of multipliers.
\newblock \emph{Foundations and Trends in Machine Learning}, 3\penalty0 (1),
  2011.

\bibitem[Chakrabarti et~al.(2001)Chakrabarti, Shi, Wirth, and
  Yao]{Chakrabarti2001}
A.~Chakrabarti, Y.~Shi, A.~Wirth, and A.~Yao.
\newblock Informational complexity and the direct sum problem for simultaneous
  message complexity.
\newblock In \emph{42nd Annual IEEE Symposium on Foundations of Computer
  Science}, pages 270--278, 2001.

\bibitem[Cover and Thomas(2006)]{Cover2006}
T.~M. Cover and J.~A. Thomas.
\newblock \emph{Elements of Information Theory, Second Edition}.
\newblock Wiley, 2006.

\bibitem[Dekel et~al.(2012)Dekel, Gilad-Bachrach, Shamir, and Xiao]{Dekel2012}
O.~Dekel, R.~Gilad-Bachrach, O.~Shamir, and L.~Xiao.
\newblock Optimal distributed online prediction using mini-batches.
\newblock \emph{Journal of Machine Learning Research}, 13:\penalty0 165--202,
  2012.

\bibitem[Duchi and Wainwright(2013)]{Duchi2013b}
J.~C. Duchi and M.~J. Wainwright.
\newblock Distance-based and continuum {F}ano inequalities with applications to
  statistical estimation.
\newblock \emph{arXiv:1311.2669 [cs.IT]}, 2013.
\newblock URL \url{http://arxiv.org/abs/1311.2669}.

\bibitem[Duchi et~al.(2012)Duchi, Agarwal, and Wainwright]{Duchi2012}
J.~C. Duchi, A.~Agarwal, and M.~J. Wainwright.
\newblock Dual averaging for distributed optimization: convergence analysis and
  network scaling.
\newblock \emph{IEEE Transactions on Automatic Control}, 57\penalty0
  (3):\penalty0 592--606, 2012.

\bibitem[Duchi et~al.(2013)Duchi, Jordan, and Wainwright]{Duchi2013}
J.~C. Duchi, M.~I. Jordan, and M.~J. Wainwright.
\newblock Local privacy and statistical minimax rates.
\newblock \emph{arXiv:1302.3203 [math.ST]}, 2013.
\newblock URL \url{http://arXiv.org/abs/1302.3203}.

\bibitem[{El Gamal} and Kim(2011)]{ElGamal2011}
A.~{El Gamal} and Y.-H. Kim.
\newblock \emph{Network Information Theory}.
\newblock Cambridge University Press, 2011.

\bibitem[Fuller and Millett(2011)]{Fuller2011}
S.~Fuller and L.~Millett.
\newblock \emph{The Future of Computing Performance: Game Over or Next Level?}
\newblock National Academies Press, 2011.

\bibitem[Han and Amari(1998)]{Han1998}
S.~Han and S.~Amari.
\newblock Statistical inference under multiterminal data compression.
\newblock \emph{IEEE Transactions on Information Theory}, 44\penalty0
  (6):\penalty0 2300--2324, 1998.

\bibitem[Hastie and Tibshirani(1995)]{Hastie1995}
T.~Hastie and R.~Tibshirani.
\newblock \emph{Generalized additive models}.
\newblock Chapman \& Hall, 1995.

\bibitem[Ibragimov and Has'minskii(1981)]{Ibragimov1981}
I.~A. Ibragimov and R.~Z. Has'minskii.
\newblock \emph{Statistical Estimation: Asymptotic Theory}.
\newblock Springer-Verlag, 1981.

\bibitem[Kushilevitz and Nisan(1997)]{Kushilevitz1997}
E.~Kushilevitz and N.~Nisan.
\newblock \emph{Communication Complexity}.
\newblock Cambridge University Press, 1997.

\bibitem[{Le Cam}(1973)]{LeCam1973}
L.~{Le Cam}.
\newblock Convergence of estimates under dimensionality restrictions.
\newblock \emph{Annals of Statistics}, 1\penalty0 (1):\penalty0 38--53, 1973.

\bibitem[Lehmann and Casella(1998)]{Lehmann1998}
E.~L. Lehmann and G.~Casella.
\newblock \emph{Theory of Point Estimation, Second Edition}.
\newblock Springer, 1998.

\bibitem[Luo(2005)]{Luo2005}
Z.-Q. Luo.
\newblock Universal decentralized estimation in a bandwidth constrained sensor
  network.
\newblock \emph{IEEE Transactions on Information Theory}, 51\penalty0
  (6):\penalty0 2210--2219, 2005.

\bibitem[Luo and Tsitsiklis(1993)]{Luo1993}
Z.-Q. Luo and J.~N. Tsitsiklis.
\newblock On the communication complexity of distributed algebraic computation.
\newblock \emph{Journal of the ACM (JACM)}, 40\penalty0 (5):\penalty0
  1019--1047, 1993.

\bibitem[Luo and Tsitsiklis(1994)]{Luo1994}
Z.-Q. Luo and J.~N. Tsitsiklis.
\newblock Data fusion with minimal communication.
\newblock \emph{IEEE Transactions on Information Theory}, 40\penalty0
  (5):\penalty0 1551--1563, 1994.

\bibitem[McDonald et~al.(2010)McDonald, Hall, and Mann]{McDonald2010}
R.~McDonald, K.~Hall, and G.~Mann.
\newblock Distributed training strategies for the structured perceptron.
\newblock In \emph{North American Chapter of the Association for Computational
  Linguistics (NAACL)}, 2010.

\bibitem[Olshevsky and Tsitsiklis(2009)]{Olshevsky2009}
A.~Olshevsky and J.~N. Tsitsiklis.
\newblock Convergence speed in distributed consensus and averaging.
\newblock \emph{SIAM Journal on Control and Optimization}, 48\penalty0
  (1):\penalty0 33--55, 2009.

\bibitem[Tsitsiklis(1993)]{Tsitsiklis1993}
J.~N. Tsitsiklis.
\newblock Decentralized detection.
\newblock In \emph{Advances in Signal Processing, Vol.~2}, pages 297--344. JAI
  Press, 1993.

\bibitem[Tsitsiklis and Luo(1987)]{Tsitsiklis1987}
J.~N. Tsitsiklis and Z.-Q. Luo.
\newblock Communication complexity of convex optimization.
\newblock \emph{Journal of Complexity}, 3\penalty0 (3):\penalty0 231--243,
  1987.

\bibitem[Tsybakov(2009)]{Tsybakov2009}
A.~B. Tsybakov.
\newblock \emph{Introduction to Nonparametric Estimation}.
\newblock Springer, 2009.

\bibitem[Yang and Barron(1999)]{Yang1999a}
Y.~Yang and A.~Barron.
\newblock Information-theoretic determination of minimax rates of convergence.
\newblock \emph{Annals of Statistics}, 27\penalty0 (5):\penalty0 1564--1599,
  1999.

\bibitem[Yao(1979)]{Yao1979}
A.~C.-C. Yao.
\newblock Some complexity questions related to distributive computing
  (preliminary report).
\newblock In \emph{Proceedings of the Eleventh Annual ACM Symposium on the
  Theory of Computing}, pages 209--213. ACM, 1979.

\bibitem[Yu(1997)]{Yu1997}
B.~Yu.
\newblock Assouad, {F}ano, and {L}e {C}am.
\newblock In \emph{Festschrift for Lucien Le Cam}, pages 423--435.
  Springer-Verlag, 1997.

\bibitem[Zhang et~al.(2013{\natexlab{a}})Zhang, Duchi, Jordan, and
  Wainwright]{Zhang2013}
Y.~Zhang, J.~C. Duchi, M.~I. Jordan, and M.~J. Wainwright.
\newblock Information-theoretic lower bounds for distributed statistical
  estimation with communication constraints.
\newblock In \emph{Advances in Neural Information Processing Systems 27},
  2013{\natexlab{a}}.

\bibitem[Zhang et~al.(2013{\natexlab{b}})Zhang, Duchi, and
  Wainwright]{Zhang2012}
Y.~Zhang, J.~C. Duchi, and M.~J. Wainwright.
\newblock Communication-efficient algorithms for statistical optimization.
\newblock \emph{Journal of Machine Learning Research}, 14:\penalty0 3321--3363,
  November 2013{\natexlab{b}}.

\end{thebibliography}

\end{document}